\RequirePackage{fix-cm}
\documentclass[12pt]{article} 
\usepackage{graphicx,natbib,amsmath,multirow,enumitem,amsfonts,verbatim,authblk,amsthm}
\usepackage{epstopdf,subfigure}

\newcommand{\wN}{$wN$ }
\newcommand{\Nw}{$Nw$ }

\theoremstyle{plain}
\newtheorem{theorem}{Theorem}
\newtheorem{corollary}[theorem]{Corollary}

\theoremstyle{definition}
\newtheorem{definition}{Definition}

\begin{document}

\title{The limits of weak selection and large population size in evolutionary game theory
}
%


\author[1]{Christine Sample}
\author[1,2]{Benjamin Allen}

\affil[1]{Department of Mathematics, Emmanuel College, 400 Fenway, Boston MA 02115}       
\affil[2]{Program for Evolutionary Dynamics, Harvard University, One Brattle Square, Cambridge MA 02138}

\date{}

\maketitle

\begin{abstract}
Evolutionary game theory is a mathematical approach to studying how social behaviors evolve.  In many recent works, evolutionary competition between strategies is modeled as a stochastic process in a finite population.  In this context, two limits are both mathematically convenient and biologically relevant: weak selection and large population size.  These limits can be combined in different ways, leading to potentially different results.  We consider two orderings: the \wN limit, in which weak selection is applied before the large population limit, and the \Nw limit, in which the order is reversed. Formal mathematical definitions of the \Nw and \wN limits are provided.  Applying these definitions to the Moran process of evolutionary game theory, we obtain asymptotic expressions for fixation probability and conditions for success in these limits.  We find that the asymptotic expressions for fixation probability, and the conditions for a strategy to be favored over a neutral mutation, are different in the \Nw and \wN limits.  However, the ordering of limits does not affect the conditions for one strategy to be favored over another.
\end{abstract}

\section{Introduction}
\label{intro}
Evolutionary game theory \citep{MaynardSmith2,MaynardSmith,Hofbauer1998,weibull1997evolutionary,broom2013game} is a framework for modeling the evolution of behaviors that affect others.  Interactions are represented as a game, and game payoffs are linked to reproductive success.  Originally formulated for infinitely large, well-mixed populations, the theory has been extended to populations of finite size \citep{TaylorFiniteGame,NowakFinite,imhof2006evolutionary,LessardFixation} and a wide variety of structures \citep{NowakMay,blume1993statistical,Ohtsuki,Corina,NowakStructured,allen2014games}.  

Calculating evolutionary dynamics in finite and/or structured populations can be difficult---in some cases, provably so \citep{ibsen2015computational}.  To obtain closed-form results, one often must pass to a limit.  Two limits in particular have emerged as both mathematically convenient and biologically relevant: large population size and weak selection.  The weak selection limit means that the game has only a small effect on reproductive success \citep{NowakFinite}.  With these limits, many results become expressible in closed form that would not be otherwise.

Often one is interested in combining these limits.  However, a central theme in mathematical analysis is that limits can be combined in (infinitely) many ways.  It is therefore important, when applying the large-population and weak-selection limits, to be clear how they are being combined.  As a first step, \cite{jeong2014optional} introduced the terms \emph{$Nw$ limit} and \emph{$wN$ limit}.  In the \Nw limit, the large population limit is taken before the weak selection limit, while in the \wN limit the order is reversed.  Informally, in the \Nw limit, the population becomes large ``much faster" than selection becomes weak, while the reverse is true for the \wN limit.  While there are infinitely many ways of combining the large-population and weak-selection limits, the $Nw$ and $wN$ limits represent two extremes in which one limit is taken entirely before the other.

Here we provide formal mathematical definitions of the $wN$ and $Nw$ limits, which were lacking in the work of \cite{jeong2014optional}.  We then apply these limits to the Moran process in evolutionary game theory \citep{Moran,TaylorFiniteGame,NowakFinite}.  We obtain asymptotic expressions for fixation probability under these limits, and show how these expressions differ depending on the order in which limits are taken.  We also analyze criteria for evolutionary success under these limits.  Our results are summarized in Table \ref{table:summary} and Figure \ref{fig:NwwNlimits}.  We show how these limits shed new light on familiar game-theoretic concepts such as evolutionary stability, risk dominance, and the one-third rule.  We also formalize and strengthen some previous results in the literature \citep{NowakFinite,antal2006fixation,bomze2008one}.

Our paper is organized as follows.  First we describe the model and define the \wN and \Nw limits. We then consider the case of constant fitness as a motivating example.  Finally, we present the results of our analysis, first for the \wN limit and then the \Nw limit.  For each limit, we derive the fixation probability for a strategy, as well as determine two conditions that measure the success of that strategy. The first condition compares the strategy's fixation probability to that of a neutral mutation.  The second compares the fixation probability of one strategy to the other.  

\begin{table}
\footnotesize
\caption{Summary of results. }
\label{table:summary}
\begin{tabular}{c l l }
\hline\noalign{\smallskip}
& \multicolumn{1}{c}{\wN limit} & \multicolumn{1}{c}{\Nw limit} \\
\noalign{\smallskip}\hline\hline\noalign{\smallskip}
	    $\rho_A$ &$\frac{1}{N}+\frac{w}{6}(a+2b-c-2d)+o(w)$&
	    $ \begin{cases}
	   o(w) &  b\le d\\
	  o(w) &  b>d \text{ and } a+b<c+d\\
	   (b-d)w +o(w) & b>d \text{ and } a+b>c+d\\
	      \frac{b-d}{2}w +o(w) & b>d \text{ and } a+b=c+d
	      \end{cases}
  $ \\\noalign{\smallskip}\hline\noalign{\smallskip}
      \multirow{2}{*}{$\rho_A>\frac{1}{N}$}& $a+2b>c+2d$, or & ($b>d$ and $a+b\ge c+d$), or\\
    &($a+2b = c+2d$ and $b>c$) &($b=d$ and $a>c$), or \\
    &&($b=d$, $a=c$ and $b>c$)
    \\\noalign{\smallskip}\hline\noalign{\smallskip}
       \multirow{2}{*}{$\rho_A>\rho_B$}&$a+b>c+d$, or& $a+b>c+d$, or\\
   &($a+b=c+d$ and $b>c$)& ($a+b=c+d$ and $b>c$)\\
\noalign{\smallskip}\hline
\end{tabular}
\normalsize
\end{table}

\begin{figure*}
	\subfigure[]{\includegraphics[width=0.5\textwidth]{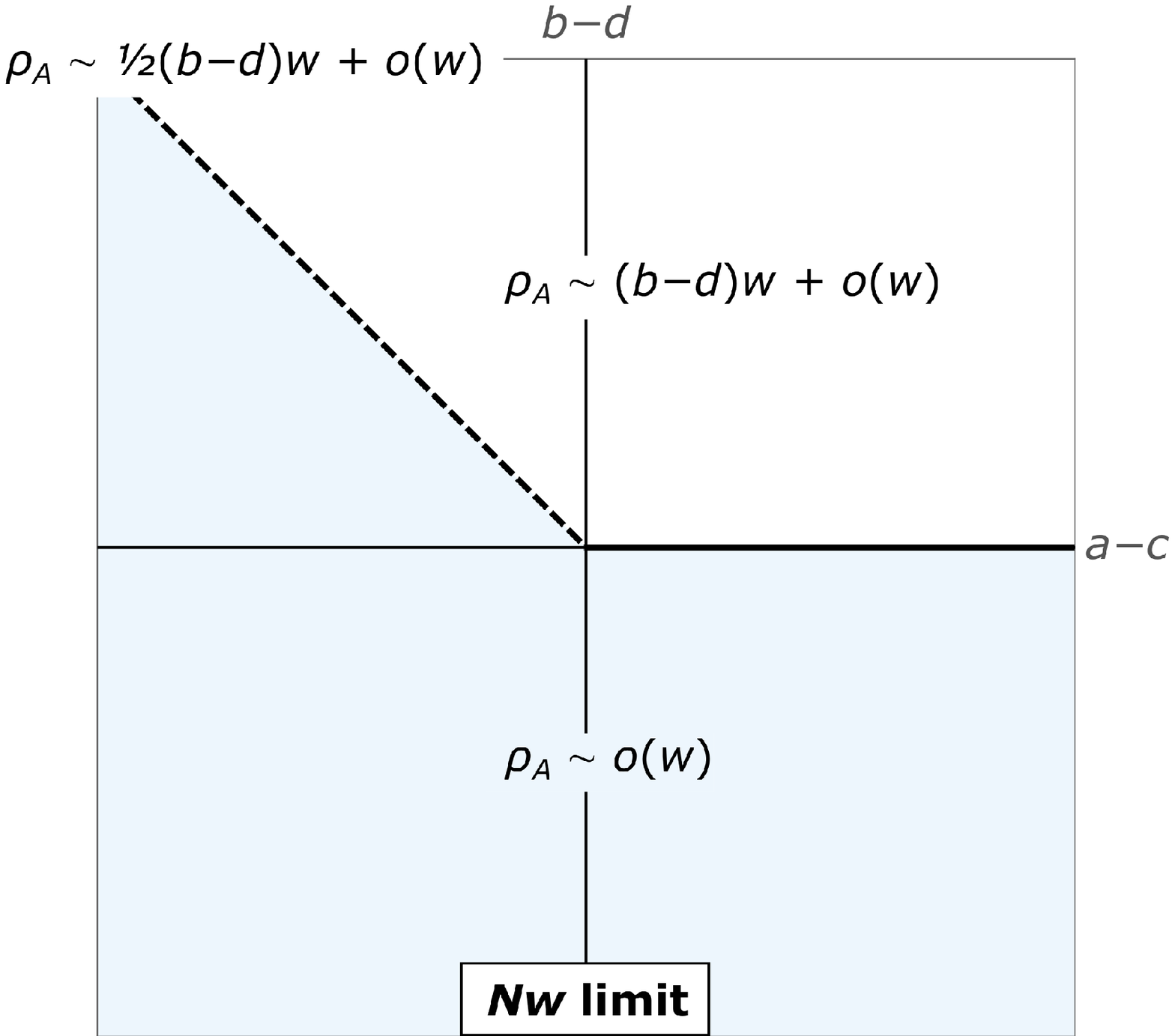}}
	\subfigure[]{\includegraphics[width=0.5\textwidth]{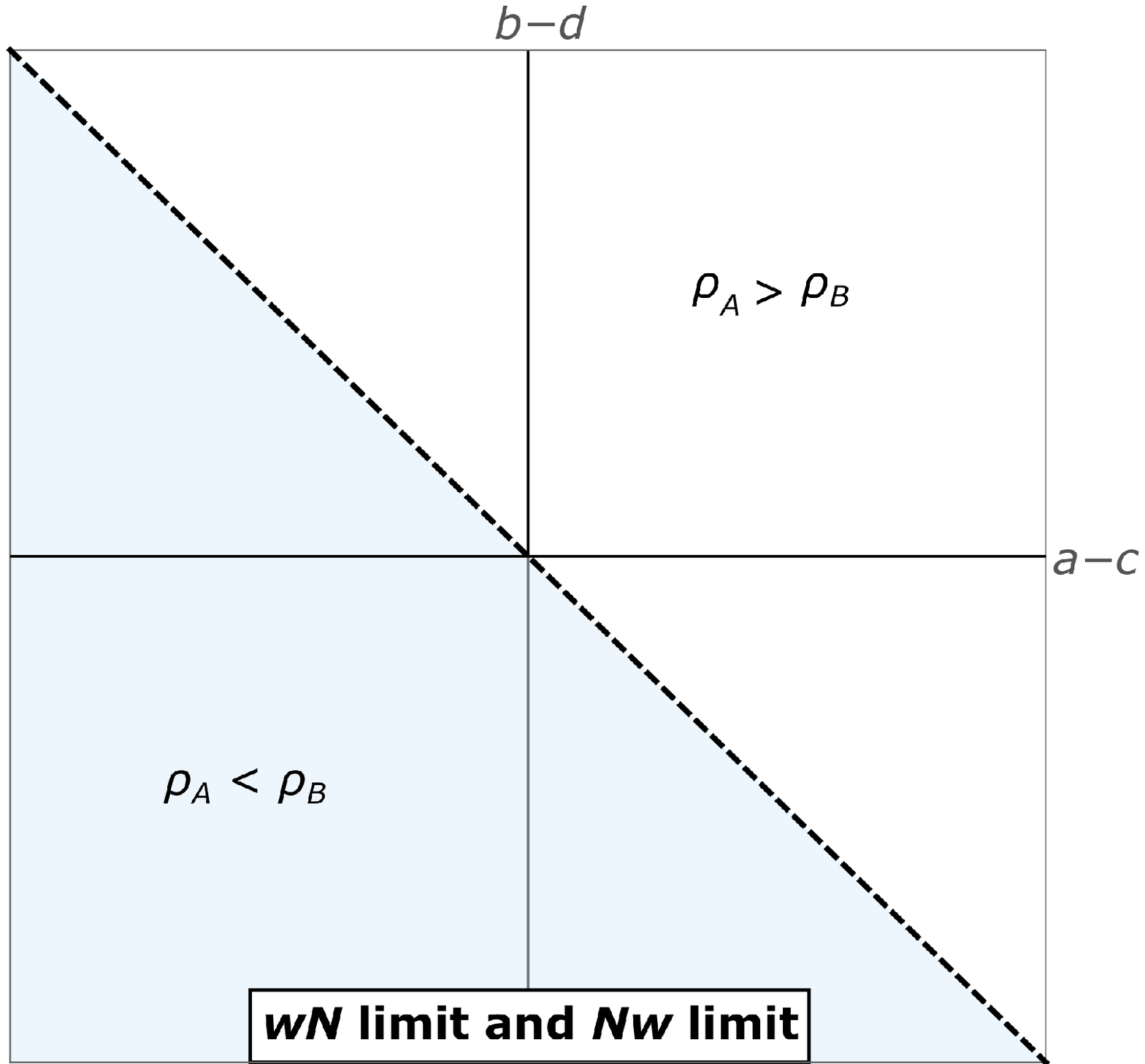}}
	\subfigure[]{\includegraphics[width=0.5\textwidth]{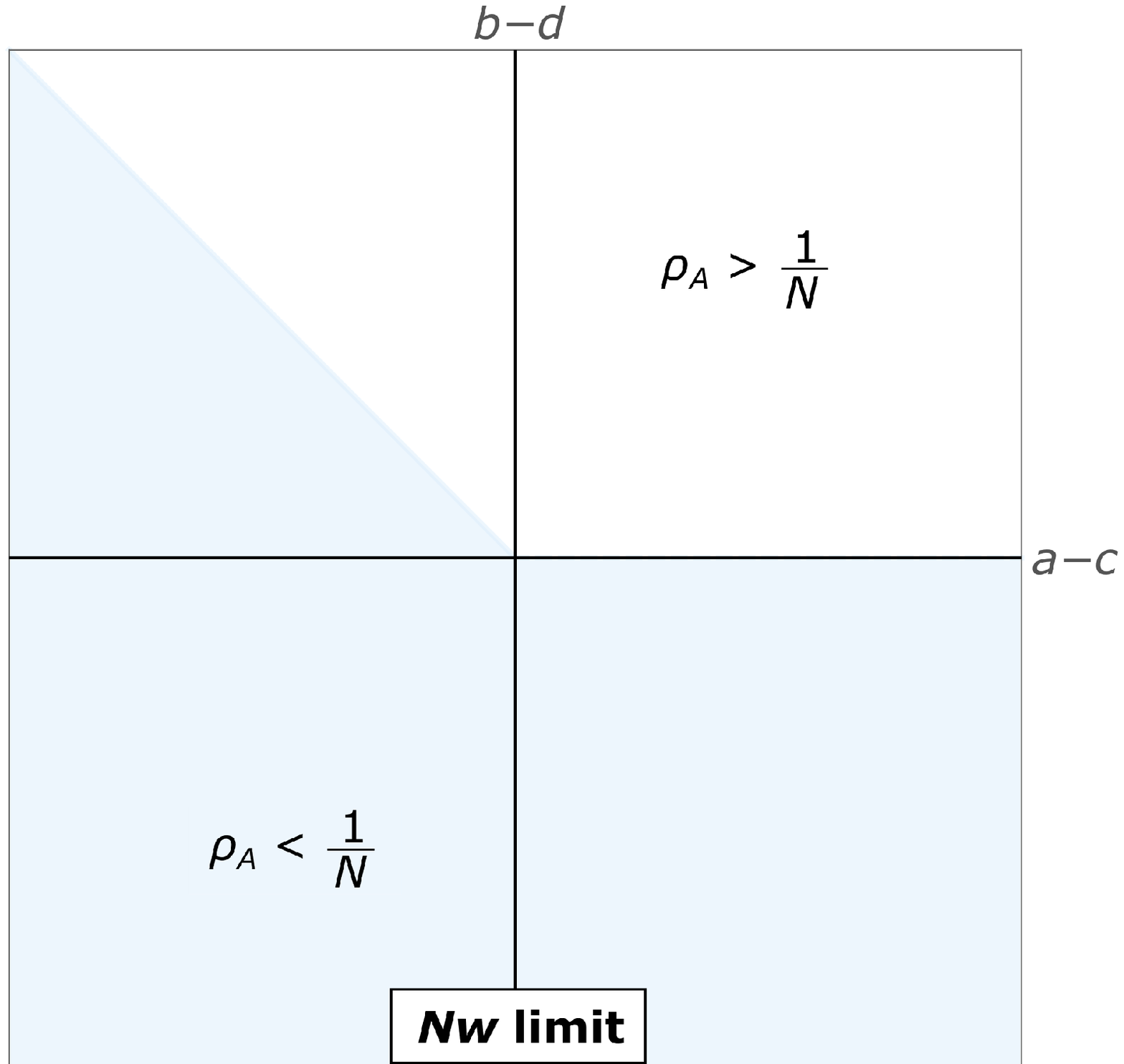}}
	\subfigure[]{\includegraphics[width=0.5\textwidth]{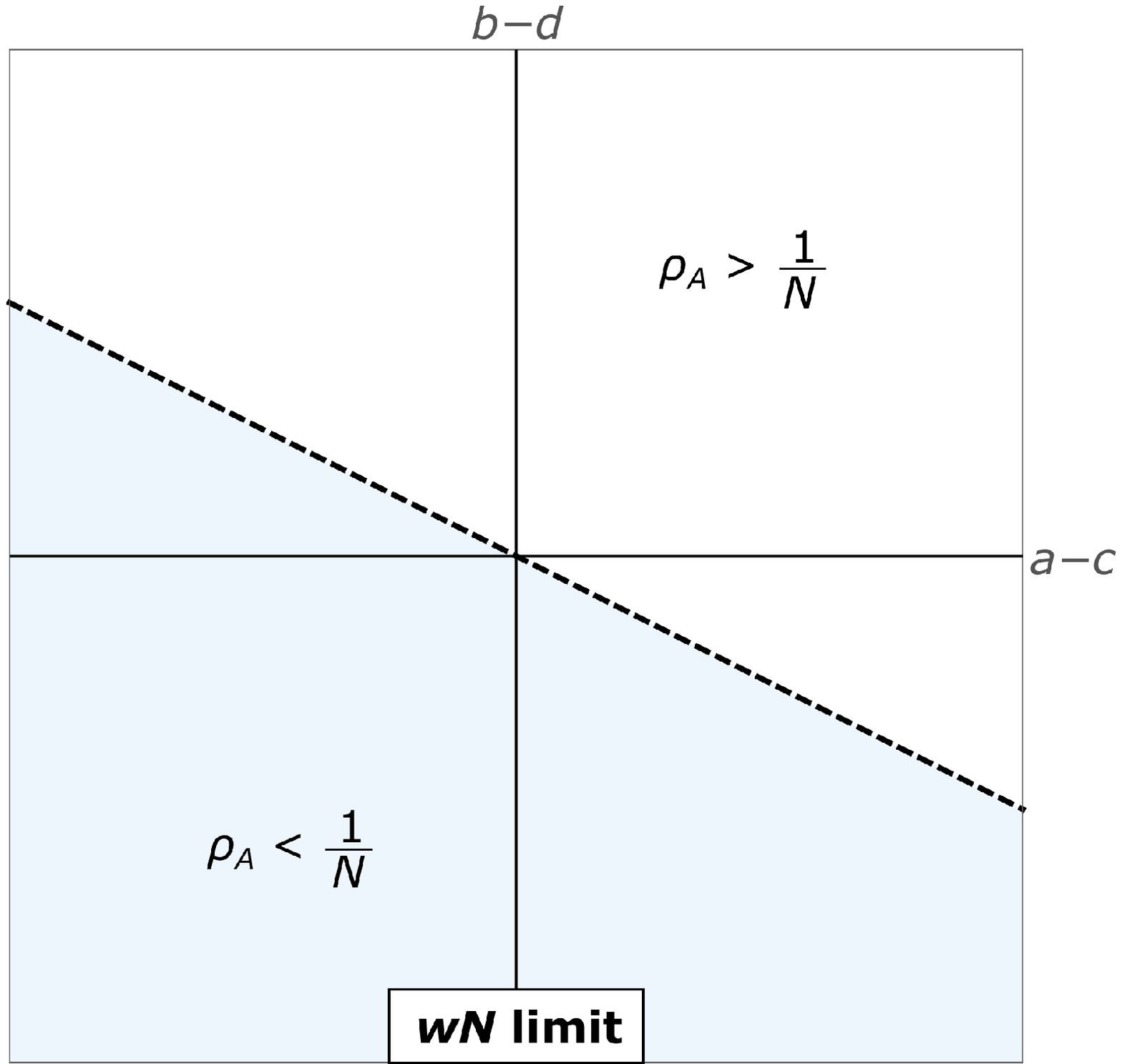}}
	\caption{Summary of our results. (a) Asymptotic expressions for $\rho_A$ under the \Nw limit in different parameter regions.  The dashed line indicates the border case $a+b=c+d$.
(b) In both the \wN limit and \Nw limit, $\rho_A > \rho_B$ if $a+b>c+d$.
(c) The order of limits matters when comparing the fixation probability of $A$ ($\rho_A$) with that of a neutral mutation ($1/N$). In the \Nw limit, $\rho_A>1/N$ if $b>d$ and $a+b \ge c+d$.
(d) In the \wN limit, $\rho_A>1/N$ if  $a+2b > c+2d$.  }
		\label{fig:NwwNlimits}     
\end{figure*}

\section{Model}
\label{sec:model}
In the Moran process \citep{Moran,TaylorFiniteGame,NowakFinite}, a population of size $N$ consists of $A$ and $B$ individuals. Interactions are described by a game
\begin{equation}\label{def:game}
	\bordermatrix{~ & A & B \cr
		A & a & b \cr
		B & c & d \cr}.
\end{equation}
The fitnesses of $A$ and $B$ individuals are defined, respectively, as expected payoffs:
\begin{align}
\label{eq:fAfB}
\begin{split}
	f_{A}(i) &= \frac{a(i-1)+b(N-i)}{N-1},\\
	f_{B}(i) &= \frac{ci+d(N-i-1)}{N-1},
	\end{split}
\end{align}
where $i$ indicates the number of $A$ individuals. Each time-step, an individual is chosen to reproduce proportionally to its fitness, and an individual is chosen with uniform probability to be replaced.  

This process has two absorbing states: $i=N$, where type $A$ has become fixed, and $i=0$, where type $B$ has become fixed.  The fixation probability of $A$, denoted $\rho_A$, is the probability that type $A$ will become fixed when starting from a state with a single $A$ individual ($i=1$).  Similarly, the fixation probability of $B$ is denoted $\rho_B$ and defined as the probability that type $B$ will become fixed  when starting from a state with single $B$ individual ($i=N-1$).  The fixation probability of $A$ can be calculated as \citep{TaylorFiniteGame}
\begin{equation}\label{eq:Moran_general}
	\rho_A = \frac{1}{1+\sum_{k=1}^{N-1}\prod_{j=1}^k \frac{f_{B}(j)}{f_{A}(j)}}.
\end{equation}
The ratio of fixation probabilities is given by
\begin{equation}\label{eq:Moran_ratio}
	\frac{\rho_A}{\rho_B} = \prod_{j=1}^{N-1} \frac{f_{A}(j)}{f_{B}(j)}.
\end{equation}

Weak selection is introduced via the following transformation of the payoff matrix:
\begin{equation}
\label{eq:weakdef}
\begin{pmatrix} a & b \\ c & d \end{pmatrix}
\mapsto \begin{pmatrix} 1+wa & 1+wb \\ 1+wc & 1+wd \end{pmatrix}.
\end{equation}
The parameter $w>0$ quantifies the strength of selection.  A result is said to hold under weak selection if it holds to first order in $w$ as $w \to 0$ \citep{NowakFinite}.

The success of strategy $A$ is quantified in two ways \citep{NowakFinite}. The first, $\rho_A > 1/N$, is the condition that selection will favor strategy $A$ over a neutral mutation (a type for which all payoff matrix entries are equal to 1).
The second condition compares the two fixation probabilities.  If $\rho_A > \rho_B$, we say that strategy $A$ is favored over strategy $B$.  

\section{Limit Definitions}
\label{sec:definitions}
We provide here formal mathematical definitions of the \wN limit, in which the weak selection is applied prior to taking the large population limit, and the \Nw limit, in which these are reversed.  We define what it means for a statement to hold true, as well as for a function to have a particular asymptotic expansion, in each of these limits.

\begin{definition}
\label{def:wN}
Statement $S(N,w)$ is \emph{True in the \wN limit} if
$$(\exists N^* \in \mathbb{N}).( \forall N \ge N^*).(\exists w^*>0). (\forall w, \,0<w<w^*).( S(N,w) \text{ is True).}$$
\end{definition}

\begin{definition}
\label{def:wN_totes}
For functions $f(N,w)$ and $g(N,w)$, we say that \emph{$f(N,w)\sim g(N,w)+o(w)$  in the \wN limit} if and only if 
$$f(N,w)=g(N,w)+wR(N,w),$$
where $ \lim_{N\rightarrow\infty} \lim_{w\rightarrow 0} R(N,w) = 0$.
\end{definition}

\begin{definition}
\label{def:Nw}
Statement $S(N,w)$ is \emph{True in the \Nw limit} if
$$(\exists w^*>0). (\forall w, \,0<w<w^*).(\exists N^* \in \mathbb{N}). (\forall N \ge N^*). (S(N,w) \text{ is True}).$$
\end{definition}

\begin{definition}
\label{def:Nw_totes}
For functions $f(N,w)$ and $g(N,w)$, we say that \emph{$f(N,w)\sim g(N,w)+o(w)$  in the \Nw limit} if and only if 
$$f(N,w)=g(N,w)+wR(N,w),$$
where $\lim_{w\rightarrow 0} \lim_{N\rightarrow\infty} R(N,w) = 0$.
\end{definition}

\section{Example: Constant Fitness}

We illustrate the difference between the \Nw and \wN limits using the special case of constant fitness.  In this case, the payoffs to $A$ and $B$ are set to constant values $f_A = 1+s$ and $f_B = 1$, independent of the population state $i$, where $s>-1$ is the selection coefficient of $A$.  The fixation probability of $A$ is \citep{Moran}
\begin{equation}
\label{eq:rho_const}
\rho_A=\frac{1-(1+s)^{-1}}{1-(1+s)^{-N}}.
\end{equation}

In the limits of large population size ($N\rightarrow\infty$) and weak selection ($s\rightarrow 0$), the asymptotic expansion of $\rho_A$ is different depending on the order in which the limits are taken (Figure ~\ref{fig:const_selection}). (Note that in the constant-fitness case, selection strength can be quantified by $|s|$ rather than $w$.) In the \wN limit, we have
\[
\rho_A \sim \frac{1}{N}+\frac{s}{2}+o(s),
\] 
whereas in the \Nw limit, 
\begin{align*}
\rho_A \sim 
\begin{cases}
0 & \text{if } s\le 0\\
s+o(s) & \text{if } s>0.
\end{cases}
\end{align*}

Although the asymptotic expressions for fixation probability differ under the two limit orderings, the conditions for success are the same. This is because, for any $s>-1$ and $N\geq 2$, type $A$ is favored over a neutral mutation ($\rho_A > 1/N$), according to  Eq.~\eqref{eq:rho_const}, if and only if $s > 0$. Likewise, $A$ is favored over $B$ ($\rho_A > \rho_B$) if and only if $s > 0$.  Since these conditions apply to arbitrary $s$ and $N$, they remain valid under any limits of these parameters.

\begin{figure*}
\begin{center}
	\subfigure[]{\includegraphics[width=0.6\textwidth]{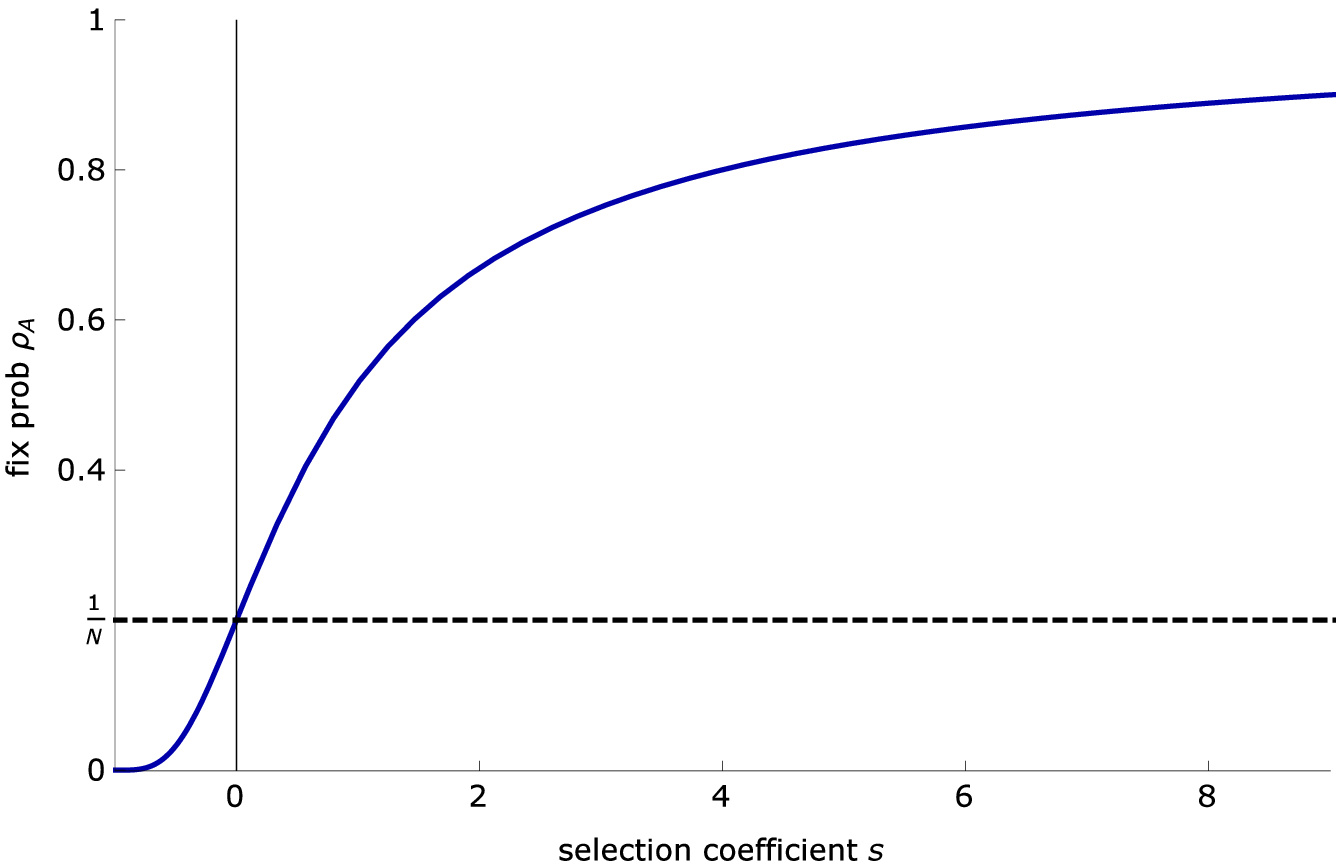}}\\
	\subfigure[]{\includegraphics[width=0.3\textwidth]{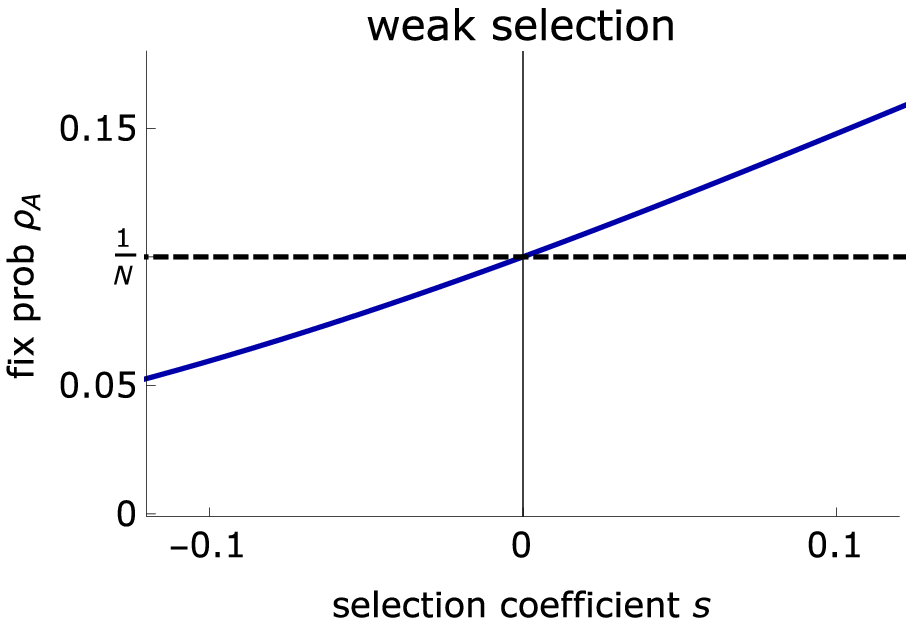}}
	\subfigure[]{\includegraphics[width=0.3\textwidth]{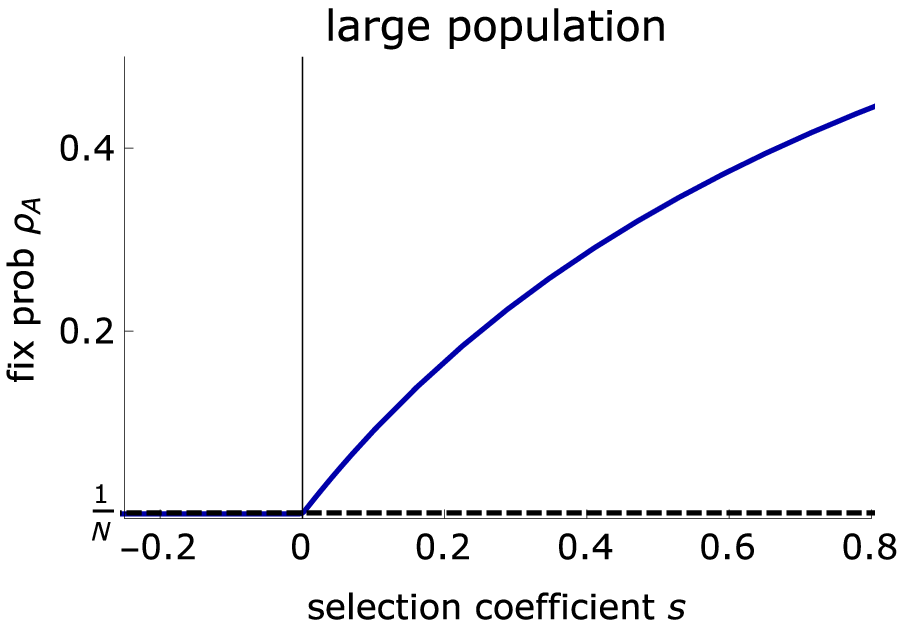}}\\
	\subfigure[]{\includegraphics[width=0.3\textwidth]{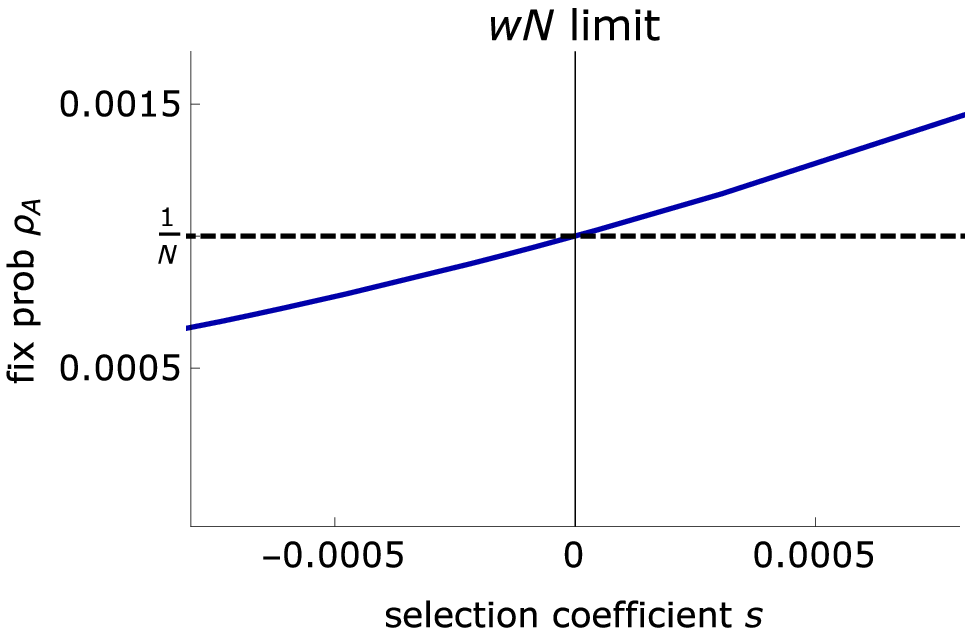}}
	\subfigure[]{\includegraphics[width=0.3\textwidth]{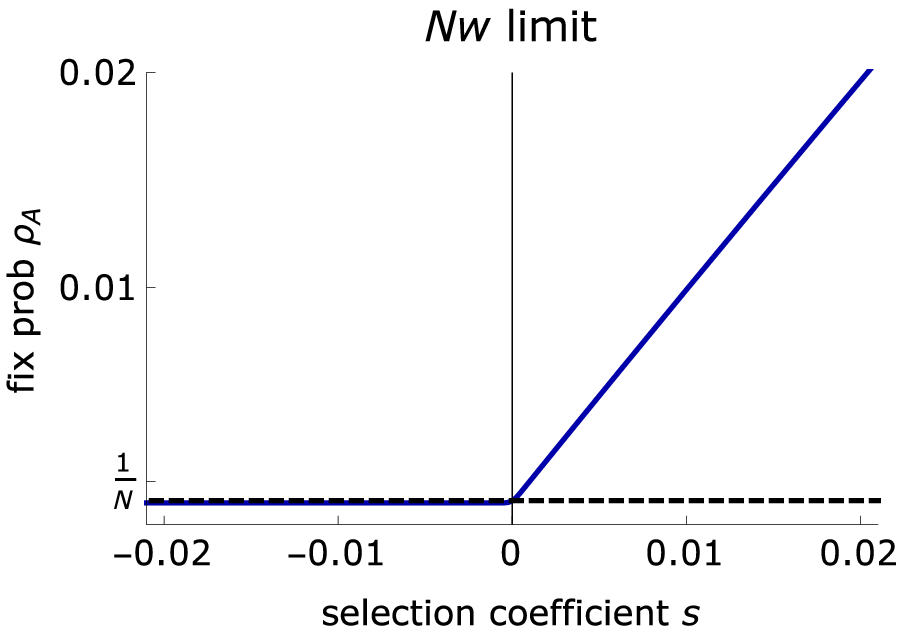}}
	\caption{Fixation probability vs.~selection coefficient for constant selection. (a) Fixation probability $\rho_A$, given by Eq.~\eqref{eq:rho_const}, is an increasing function of the selection coefficient $s$. (b) When selection is weak ($|s|\ll 1$), fixation probability is approximately linear in $s$.  (c) For large population size ($N \to \infty$), fixation probability goes to zero for $s \leq 0$, and there is a corner in the graph at $s=0$.  (d) In the \wN limit, weak selection is applied first followed by large population size, resulting in $\rho_A \sim 1/N+s/2+o(s)$. (e) In the \Nw limit, the limit $N\rightarrow\infty$ is applied first followed by weak selection.  The result is a piecewise-linear function which is zero for $s \leq 0$ and has slope 1 for $s>0$. Population size is $N=5, 10, 10^3, 10^3,$ and $10^4$ in panels (a)-(e), respectively. }
	\label{fig:const_selection}   
\end{center}
\end{figure*}

\section{Results}
\label{sec:results}

Having motivated our investigation using the case of constant selection, we now consider an arbitrary payoff matrix \eqref{def:game}.  We analyze the \wN limit first, followed by the \Nw limit.

\subsection{\wN Limit}
\label{sec:wN}
In the \wN limit we first apply weak selection and then consider large population size.  Results for $\rho_A$ are presented first, followed by conditions for success.
\begin{theorem}  In the \wN limit, $\rho_A \sim \frac{1}{N}+\frac{w}{6}(a+2b-c-2d)+o(w)$.
\end{theorem}

This theorem formalizes a result of \cite{NowakFinite}.

\begin{proof}
We apply weak selection to the fitnesses in Eq.~\eqref{eq:fAfB}:
\begin{align}
\label{eq:fAfBweak}
\begin{split}
f_{A}(i) &= 1+w\frac{a(i-1)+b(N-i)}{N-1},\\
f_{B}(i) &= 1+w\frac{ci+d(N-i-1)}{N-1}.
\end{split}
\end{align}
Substituting Eq.~\eqref{eq:fAfBweak} into \eqref{eq:Moran_general} and taking a Taylor expansion about $w=0$ gives
\begin{align}\label{eq:rhoA_weak}
\rho_A&=\frac{1}{N}+\frac{w}{6N}\left(N(a+2b-c-2d)-(2a+b+c-4d)\right)+wQ(N,w),
\end{align}
where $\lim_{w\rightarrow 0} Q(N,w) = 0$.
We regroup,
\begin{align*}
\rho_A &=\frac{1}{N}+\frac{w}{6}(a+2b-c-2d)+wR(N,w),
\end{align*}
and define the remainder term as $R(N,w) = Q(N,w)-\frac{1}{6N}(2a+b+c-4d)$.  By taking the limit of $R(N,w)$ first as $w\rightarrow 0$ then as $N\rightarrow \infty$, we find that
\begin{align*}
\lim_{N\rightarrow \infty} \lim_{w\rightarrow 0} R(N,w)&=\lim_{N\rightarrow \infty} \lim_{w\rightarrow 0}\left( Q(N,w)-\frac{1}{6N}(2a+b+c-4d)\right)\\
&=\lim_{N\rightarrow \infty} \left(-\frac{1}{6N}(2a+b+c-4d)\right)\\
&= 0.
\end{align*}
By Definition~\ref{def:wN_totes}, $\rho_A \sim \frac{1}{N}+\frac{w}{6}(a+2b-c-2d)+o(w)$ in the \wN limit. \qed
\end{proof}

\subsubsection{Conditions for Success}
\label{sec:wNsuccess}
\begin{theorem}
\label{thm:wN_NrhoA}
\renewcommand{\labelenumi}{(\roman{enumi})}
In the \wN limit, $\rho_A > \frac{1}{N}$ if and only if one of the following holds:
\begin{enumerate}
	\item $a+2b>c+2d$
	\item $a+2b=c+2d$ and $b > c$.
\end{enumerate}
\end{theorem}

An equivalent result was obtained by \cite{bomze2008one}.

\begin{proof}
 Under weak selection, it is apparent from Eq.~\eqref{eq:rhoA_weak}  that $\rho_A > 1/N$ 
if $N(a+2b-c-2d)-(2a+b+c-4d)>0$ and $\rho_A < 1/N$ if $N(a+2b-c-2d)-(2a+b+c-4d)<0$.  Thus $\rho_A > 1/N$ for sufficiently large $N$ if $a+2b>c+2d$ or if $a+2b=c+2d$ and $2a+b+c-4d<0$.  The second condition is equivalent to  $a+2b=c+2d$ and $b > c$.

For the border case, $a+2b=c+2d$ and $b = c$, we take a second-order expansion of $\rho_A$:
\begin{align*}
\rho_A&=\frac{1}{N}-w^2\frac{(a-b)^2(N+2)(N+1)(N-2)}{240N(N-1)}+\mathcal{O}(w^3).
\end{align*}
For $N>2$ and $a\ne b$, the second order term is always negative, which implies that $\rho_A < 1/N$.  Lastly, if $a = b = c = d$ then $\rho_A = 1/N$.\qed
\end{proof}

\begin{theorem}
\renewcommand{\labelenumi}{(\roman{enumi})}  
In the \wN limit, $\rho_A > \rho_B$ if and only if one of the following holds:
\begin{enumerate}
	\item $a+b>c+d$
	\item $a+b=c+d$ and $b > c$.
\end{enumerate}
\end{theorem}

Case (i) of this result was stated informally by \cite{NowakFinite}.

\begin{proof}
 Substituting Eq.~\eqref{eq:fAfBweak}  into Eq.~\eqref{eq:Moran_ratio} and taking a Taylor expansion about $w=0$, we get 
\begin{align*}
\frac{\rho_A}{\rho_B} &= \prod_{j=1}^{N-1} \frac{N-1+w\left(a(j-1)+b(N-j)\right)}{N-1+w\left(cj+d(N-j-1)\right)}\\
&=1+\frac{w}{2}\left(N(a+b-c-d)-2a+2d\right)+wQ(N,w),
\end{align*}
where $\lim_{w\rightarrow 0} Q(N,w) = 0$.
Clearly, $\rho_A$ is greater than (less than) $\rho_B$ under weak selection if $N(a+b-c-d)-2a+2d$ is positive (negative). The expression is positive for sufficiently large $N$ if $a+b>c+d$ or if 
$a+b=c+d$ and $a<d$.  The second condition is equivalent to  $a+b=c+d$ and $b > c$.  Lastly, if $b=c$ and $a=d$, then from  Eq.~\eqref{eq:Moran_ratio},  $\rho_A = \rho_B$.\qed
\end{proof}

\subsection{\Nw Limit}
\label{sec:Nw}
In this section, we first determine the limit of $\rho_A$ as $N\rightarrow \infty$ (Theorem \ref{thm:rhoA_Nlimit}) before finding an asymptotic expression for $\rho_A$ in the \Nw limit.  We then turn to conditions for success, first in the $N \to \infty$ limit (Theorems \ref{thm:N_NrhoA} and \ref{thm:ratio_largeN}) and then the \Nw limit.  

\begin{theorem} \label{thm:rhoA_Nlimit}
The fixation probability $\rho_A$ has the following large-population limit:
\begin{align}\label{eq:rhoA_Nlimit}
\lim_{N\rightarrow \infty}\rho_A =
  \begin{cases}
   0 & \text{if } b\le d\, \\
   0 & \text{if } b>d, \, a<c \text{ and } I>0\\
   \frac{(b-d)(c-a)}{b(c-a)+c(b-d)\sqrt{\frac{ac}{bd}}} & \text{if } b>d, \, a< c \text{ and } I=0 \\   
   \frac{b-d}{b} & \text{if } b>d, \, a< c \text{ and } I<0 \\
      \frac{b-d}{b} & \text{if } b>d, \,a\ge c,
  \end{cases}
 \end{align}
where 
\begin{align}\label{def:I}
	I = \int_0^1\ln\tilde{f}(x)dx,
\end{align} 
and
\begin{equation}
\tilde{f}\left(x\right) = \frac{d+x(c-d)}{b+x(a-b)} \text{ for } x\in[0,1]. \label{def:ftilde}
\end{equation}
 \end{theorem}

Some aspects of this result were obtained by \cite{antal2006fixation}.  However, their derivations used approximations that require formal verification.  Our proof confirms most of the results of \cite{antal2006fixation} but contradicts their result in the case $b>d$, $a< c$, and $I=0$, as we detail in the Discussion.

 \begin{proof}
We first establish some basic results before considering various cases.  
From Eq.~\eqref{eq:fAfB}, define the function $f\left(\frac{i}{N},N\right)$ as
 \begin{align}
f\left(\frac{i}{N},N\right) &= \frac{f_{B}(i)}{f_{A}(i)} = \frac{d+\frac{i}{N}(c-d)-\frac{d}{N}}{b+\frac{i}{N}(a-b)-\frac{a}{N}}.\label{def:f}
\end{align}
$\tilde{f}\left(\frac{i}{N}\right)$ of Eq.~\eqref{def:ftilde} serves as an approximation to $f\left(\frac{i}{N},N\right)$ with error:
\begin{align*}
\epsilon_N(i) &= f\left(\frac{i}{N},N\right) - \tilde{f}\left(\frac{i}{N}\right)\nonumber\\
&=\frac{1}{N}\cdot\frac{ad-bd-\frac{i}{N}\left(2ad-bd-ac\right)}{\left[b+\frac{i}{N}\left(a-b\right)\right]\left[b+\frac{i}{N}\left(a-b\right)-\frac{a}{N}\right]}.
\end{align*}
Importantly, $\epsilon_N(i)$ is uniformly bounded in the sense that, for $N$ sufficiently large, there exists a positive constant $L$ such that $|\epsilon_N(i)|\le \frac{L}{N}$ for all $i=1,...,N$.  
Specifically, for $N\ge \frac{2a}{\min\left\{a,b\right\}}$, we can set 
\begin{align*}
L = \frac{2\max\left\{|ad-bd|,|ac-ad|\right\}}{\left(\min\left\{a,b\right\}\right)^2}.
\end{align*}
Therefore, $\lim_{N\rightarrow\infty} f(x,N) = \tilde{f}(x) \text{ uniformly in } x$.

The function $\tilde{f}(x)$ has some useful properties.  For instance, if $bc= ad$ then $\tilde{f}(x) = d/b$ is a constant function. Otherwise,
$\tilde{f}(x)$ is  monotonic: the derivative
\begin{align*}
\frac{d\tilde{f}}{d x}
&= \frac{bc-ad}{\left(b+x(a-b)\right)^2}
\end{align*}
implies that $\tilde{f}(x)$ is always strictly increasing ($bc>ad$) or strictly decreasing ($bc<ad$). Extrema must occur at the endpoints  $\tilde{f}(0)=d/b$ and $\tilde{f}(1)=c/a$. Set
\begin{align}\label{def:mtilde}
\begin{split}
\tilde{m}&=\min\left\{\tilde{f}(0),\tilde{f}(1)\right\}\\
\tilde{M}&=\max\left\{\tilde{f}(0),\tilde{f}(1)\right\}.
\end{split}
\end{align}  

Our proof makes frequent use of the integral $I$ of Eq.~\eqref{def:I},
which is evaluated as:
\begin{align}
I &= \ln\left(\frac{b^{\frac{b}{a-b}}c^{\frac{c}{c-d}}}{a^{\frac{a}{a-b}}d^{\frac{d}{c-d}}}\right).\label{eq:int_value}
\end{align}
An illustration of this integral is given in Figure \ref{fig:integral}.

\begin{figure*}
		\includegraphics[width=\textwidth]{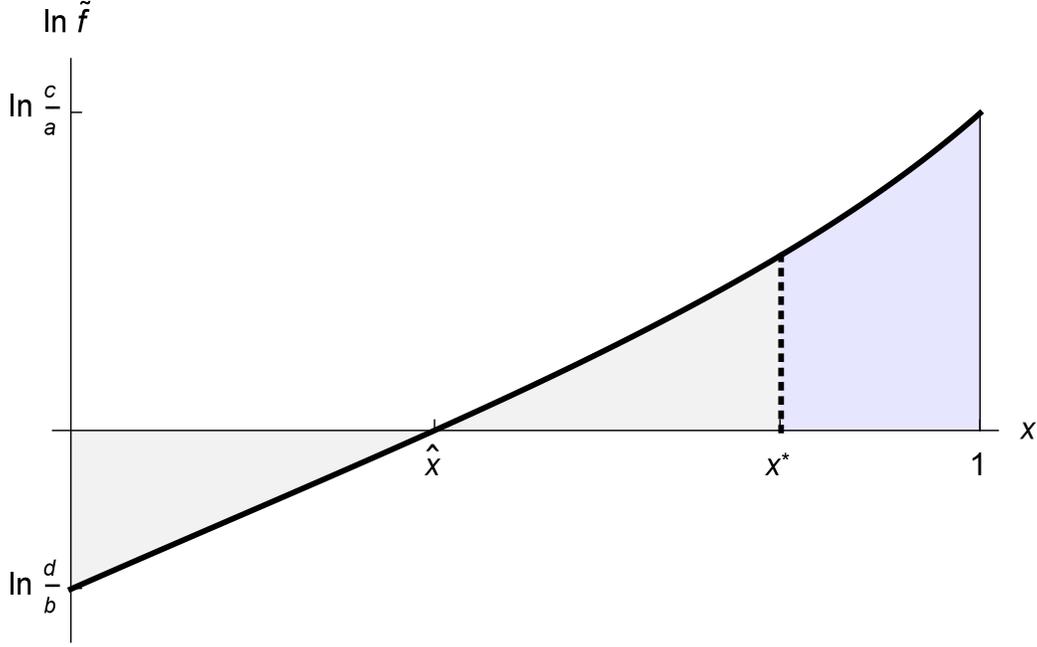}
	\caption{Plot of $\ln \tilde{f}(x)$ vs $x$, where $\tilde{f}(x)$ is defined as in Eq.~\eqref{def:ftilde}.  This figure illustrates the case that $b>d$, $c>a$ and $I>0$ (the net area under the curve is positive). The point $x^*$ satisfies $\int_0^{x^*}\ln\tilde{f}(x)\; dx = 0$. }
	\label{fig:integral}
\end{figure*}

Our objective is to investigate the fixation probability of Eq.~\eqref{eq:Moran_general}, which can be written
\begin{align}\label{eq:rhoA_S}
\rho_A = \frac{1}{1+S},
\end{align}
where $S$ is the sum defined as
 \begin{align}\label{def:S}
S = \sum_{k=1}^{N-1}\prod_{i=1}^k f\left(\frac{i}{N},N\right).
\end{align}
The product in Eq.~\eqref{def:S} can be written as $\prod_{i=1}^k \left[\tilde{f}\left(\frac{i}{N}\right)+\epsilon_N(i)\right]$.
The bound on $\epsilon_N(i)$ implies that for sufficiently large $N$,
\begin{align*}
\prod_{i=1}^k \left[\tilde{f}\left(\frac{i}{N}\right)-\frac{L}{N}\right]&\le\prod_{i=1}^k f\left(\frac{i}{N},N\right)
\le\prod_{i=1}^k \left[\tilde{f}\left(\frac{i}{N}\right)+\frac{L}{N}\right].
\end{align*}
Using $\tilde{m}$, the minimum of $\tilde{f}$ given in Eq.~\eqref{def:mtilde}, we obtain 
\begin{align}
\left(1-\frac{L}{\tilde{m}N}\right)^k\prod_{i=1}^k \tilde{f}\left(\frac{i}{N}\right)&\le\prod_{i=1}^k f\left(\frac{i}{N},N\right)\le\left(1+\frac{L}{\tilde{m}N}\right)^k\prod_{i=1}^k \tilde{f}\left(\frac{i}{N}\right).\label{ineq:ftilde_prod}
\end{align}
These inequalities allow for the comparison between $f$ and $\tilde{f}$.

We now split the sum of Eq.~\eqref{def:S} as $S = S_1+S_2$, where $S_1$ and $S_2$ are non-negative sums defined as
\begin{align}
S_1&=\sum_{k=1}^{\lfloor \ln{N}\rfloor} \prod_{i=1}^k f\left(\frac{i}{N},N\right),\label{def:S1}\\
S_2&=\sum_{k=\lfloor \ln{N}\rfloor+1}^{N-1}
\prod_{i=1}^k f\left(\frac{i}{N},N\right).
\label{def:S2}
\end{align}

Let
\begin{align}
\label{def:mM12}
\begin{split}
m_1 &= \min_{1\le i \le \lfloor \ln{N}\rfloor}\left\{f\left(\frac{i}{N},N\right)\right\},\\
M_1 &= \max_{1\le i \le \lfloor \ln{N}\rfloor}\left\{f\left(\frac{i}{N},N\right)\right\},\\
m_2 &= \min_{\lfloor \ln{N}\rfloor+1\le i \le N-1}\left\{f\left(\frac{i}{N},N\right)\right\},\\
M_2 &= \max_{\lfloor \ln{N}\rfloor+1\le i \le N-1}\left\{f\left(\frac{i}{N},N\right)\right\}.\\
\end{split}
\end{align}

Since $f$ converges uniformly to the monotonic function $\tilde{f}$,
\begin{align}
\label{def:mM12_limit}
\begin{split}
\lim_{N\rightarrow\infty} m_1 &= \tilde{f}(0) = \frac{d}{b},\\
\lim_{N\rightarrow\infty} M_1 &= \tilde{f}(0) = \frac{d}{b},\\
\lim_{N\rightarrow\infty} m_2 &= \tilde{m}=\min\left\{\frac{d}{b},\frac{c}{a}\right\},\\
\lim_{N\rightarrow\infty} M_2 &= \tilde{M}=\max\left\{\frac{d}{b},\frac{c}{a}\right\}.\\
\end{split}
\end{align}

Useful inequalities obtained from Eqs.~\eqref{def:S1} and \eqref{def:mM12} are
\begin{align}\label{ineq:S1_minsum}
\sum_{k=1}^{\lfloor \ln{N}\rfloor} \prod_{i=1}^k m_1 &\le S_1 \le  \sum_{k=1}^{\lfloor \ln{N}\rfloor} \prod_{i=1}^k M_1, \nonumber\\
\sum_{k=1}^{\lfloor \ln{N}\rfloor} m_1^k &\le S_1 \le \sum_{k=1}^{\lfloor \ln{N}\rfloor} M_1^k.
\end{align}
The geometric series gives
\begin{align}
\frac{m_1-m_1^{\lfloor \ln{N}\rfloor+1}}{1-m_1}&\le S_1 \le \frac{M_1-M_1^{\lfloor \ln{N}\rfloor+1}}{1-M_1},\label{ineq:S1_min_max}
\end{align}
as long as $m_1\ne 1$ and $M_1\ne 1$, respectively.

Now to determine $\lim_{N\rightarrow\infty}\rho_A$, we consider cases. We first compare $b$ and $d$.  If necessary, we then compare $a$ and $c$ and if further required, consider the sign of $I$.

\begin{enumerate}
\item{Case $b<d$

In this case,  $\lim_{N\rightarrow\infty} m_1 = d/b >1$ and
\begin{align*}
\lim_{N\rightarrow\infty}\frac{m_1-m_1^{\lfloor \ln{N}\rfloor+1}}{1-m_1}=\infty.
\end{align*}
It follows from Eq.~\eqref{ineq:S1_min_max} that $\lim_{N\rightarrow\infty}S_1=\infty$ and consequently, $\lim_{N\rightarrow\infty}S = \infty$. Eq.~\eqref{eq:rhoA_S} gives $\lim_{N\rightarrow\infty}\rho_A = 0$.
}
\item{Case $b=d$

In this case,  $\lim_{N\rightarrow\infty} m_1 = d/b =1$. Fix an arbitrary positive integer $B$ so that
\begin{align*}
\lim_{N\rightarrow\infty}\sum_{k=1}^{B+1}m_1^k =B+1.
\end{align*}
This implies that for all sufficiently large $N$,
\begin{align*}
\sum_{k=1}^{B+1}m_1^k > B.
\end{align*}
In particular, for $\lfloor \ln{N} \rfloor > B+1$, 
\begin{align*}
\sum_{k=1}^{\lfloor \ln{N}\rfloor} m_1^k>\sum_{k=1}^{B+1} m_1^k >B.
\end{align*}
Since $B$ was arbitrary $\sum_{k=1}^{\lfloor \ln{N}\rfloor} m_1^k$ becomes larger than any positive integer as $N\rightarrow\infty$.
This proves that
\begin{align*}
\lim_{N\rightarrow\infty}\sum_{k=1}^{\lfloor \ln{N} \rfloor}m_1^k=\infty.
\end{align*}
From Eq.~\eqref{ineq:S1_minsum} we conclude that $\lim_{N\rightarrow\infty}S_1 = \infty$ and consequently $\lim_{N\rightarrow\infty}\rho_A = 0$.
}

\item{Case $b>d$

Under this case 
$\lim_{N\rightarrow\infty} m_1 = \lim_{N\rightarrow\infty} M_1 = d/b <1$. From Eq.~\eqref{ineq:S1_min_max}, $S_1$ is bounded, and it follows from taking the limit as $N\rightarrow\infty$ of Eq.~\eqref{ineq:S1_min_max} and applying the Squeeze Theorem that
\begin{align}\label{eq:S1_limit}
\lim_{N\rightarrow\infty}S_1 = \frac{d}{b-d}.
\end{align}

We now turn our attention to $S_2$, which requires the consideration of subcases.
\begin{enumerate}
\item{Subcase $a > c$

Eq.~\eqref{def:mM12_limit} implies $\lim_{N\rightarrow \infty}m_2<1$ and $\lim_{N\rightarrow \infty}M_2<1$.  Furthermore, $S_2$ of Eq.~\eqref{def:S2} is bounded: 
\begin{align*}
\sum_{k=\lfloor \ln{N}\rfloor+1}^{N-1}m_2^k &\le S_2\le \sum_{k=\lfloor \ln{N}\rfloor+1}^{N-1}M_2^k \nonumber\\
\frac{{m_2}^{\lfloor \ln{N}\rfloor+1}-{m_2}^{N}}{1-{m_2}}&\le S_2\le\frac{{M_2}^{\lfloor \ln{N}\rfloor+1}-{M_2}^{N}}{1-{M_2}}.
\end{align*}
It follows from the Squeeze Theorem that
\begin{align}\label{eq:S2_limit}
\lim_{N\rightarrow\infty}S_2 = 0.
\end{align}
Eqs.~\eqref{eq:S1_limit} and \eqref{eq:S2_limit} together give
$\lim_{N\rightarrow\infty} S = d/(b-d)$
and by Eq.~\eqref{eq:rhoA_S},
$\lim_{N\rightarrow\infty} \rho_A = (b-d)/b$.
}
\item{Subcase $a < c$

In this case, $\tilde{f}$ is an increasing function with minimum value of $\tilde{f}(0) = d/b<1$ and maximum value of  $\tilde{f}(1) = c/a>1$. The behavior of $\rho_A$ depends on the sign of the integral $I$. Therefore, we must consider subcases to this subcase.  An illustration is given in Figure \ref{fig:integral} for the subcase  $I>0$. 

\begin{enumerate}
\item{\label{case:intneg}
Subcase $I<0$

We will show that $S_2\rightarrow 0$ as $N\rightarrow\infty$. Set
\begin{align}
\tilde{A}_k = \sum_{i=1}^{k} \ln\tilde{f}\left(\frac{i}{N}\right),\label{def:Ak} 
\end{align}
and
\begin{align}
\tilde{S_2} &= \sum_{k=\lfloor \ln{N}\rfloor+1}^{N-1}
\prod_{i=1}^k \tilde{f}\left(\frac{i}{N}\right)
=\sum_{k=\lfloor \ln{N}\rfloor+1}^{N-1}
\exp \tilde{A}_k. \label{eq:S2tilde}
\end{align}

We will show that $\exp \tilde{A}_k$ is less than or equal to some constant multiple of $e^{kI}$, where $I$ is defined in Eq.~\eqref{def:I}.

Consider the integral $\int_0^{k/N} \ln\tilde{f}\left(x\right)dx$.  Since $\ln\tilde{f}\left(x\right)$ is a monotonically increasing function, the left Riemann sum is a lower bound:
\begin{align}
\nonumber
\int_0^{k/N} \ln\tilde{f}\left(x\right)dx&>\frac{1}{N}\sum_{i=0}^{k-1} \ln \tilde{f}\left(\frac{i}{N}\right)\\
\label{eq:leftRiemann}
&=\frac{1}{N}\left(\tilde{A}_k+\ln \tilde{f}\left(0\right) -\ln \tilde{f}\left(\frac{k}{N}\right)\right).
\end{align}
Furthermore, the maximum value of $\ln \tilde{f}(x)$ is $\ln\tilde{f}(1)$. Substituting this bound into \eqref{eq:leftRiemann} and rearranging, we have that for all $k=1,...,N$,
\begin{align}\label{ineq:Ak_int}
 \tilde{A}_k&< N\int_0^{k/N} \ln\tilde{f}\left(x\right)dx+\ln \frac{\tilde{f}\left(1\right)}{\tilde{f}\left(0\right)}.
\end{align}

Since $\ln\tilde{f}$ is increasing, the average value of $\ln \tilde{f}(x)$ over intervals $[0,y]$ must be increasing in $y$.  Hence for $y\in[0,1]$,
\begin{align*}
\frac{1}{y}\int_0^{y} \ln\tilde{f}\left(x\right)dx \le 
\int_0^{1} \ln\tilde{f}\left(x\right)dx  = I.
\end{align*}
Let $y = k/N$ to obtain
\begin{align*}
N\int_0^{k/N} \ln\tilde{f}\left(x\right)dx \le kI.
\end{align*}
Combining with Eq.~\eqref{ineq:Ak_int},
\begin{align}\label{ineq:Ak}
\tilde{A}_k< kI +\ln \frac{\tilde{f}\left(1\right)}{\tilde{f}\left(0\right)}.
\end{align}
Substitute Eq.~\eqref{ineq:Ak} into Eq.~\eqref{eq:S2tilde} to obtain
\begin{align*}
\tilde{S}_2&<
\sum_{k=\lfloor \ln{N}\rfloor+1}^{N-1} \frac{\tilde{f}\left(1\right)}{\tilde{f}\left(0\right)} e^{kI}=\frac{\tilde{f}\left(1\right)}{\tilde{f}\left(0\right)}\cdot\frac{e^{I(\lfloor \ln{N}\rfloor+1)}-e^{IN}}{1-e^{I}}.
\end{align*}
Therefore since $I<0$,
\begin{align}
\lim_{N\rightarrow\infty}\tilde{S}_2&= 0.\label{eq:S2tilde_limit}
\end{align}

We must now show how $\tilde{S}_2$ relates to $S_2$.  Substitute $\tilde{m} = \tilde{f}(0)=d/b$ into Eq.~\eqref{ineq:ftilde_prod} and sum over $k$ to obtain an upper bound for $S_2$:
\begin{align*}
\sum_{k=\lfloor \ln{N}\rfloor+1}^{N-1}\prod_{i=1}^k f\left(\frac{i}{N},N\right)&\le\sum_{k=\lfloor \ln{N}\rfloor+1}^{N-1}\left(1+\frac{bL}{dN}\right)^k\prod_{i=1}^k \tilde{f}\left(\frac{i}{N}\right)\\
&\le\left(1+\frac{bL}{dN}\right)^N\sum_{k=\lfloor \ln{N}\rfloor+1}^{N-1}\prod_{i=1}^k \tilde{f}\left(\frac{i}{N}\right).
\end{align*}
Thus,
\begin{align*}
S_2 &\le \left(1+\frac{bL}{dN}\right)^N \tilde{S}_2.
\end{align*}
The limit
\begin{align}\label{eq:exp_limit}
\lim_{N\rightarrow\infty}\left(1+\frac{bL}{dN}\right)^N = e^{bL/d},
\end{align}
together with Eq.~\eqref{eq:S2tilde_limit} gives
\begin{align}\label{eq:S2_limit_intneg}
\lim_{N\rightarrow\infty}S_2&= 0.
\end{align}

Adding Eqs.~\eqref{eq:S1_limit} and \eqref{eq:S2_limit_intneg}, we find $\lim_{N\rightarrow\infty} S = d/(b-d)$
and consequently $\lim_{N\rightarrow\infty} \rho_A = (b-d)/b$.
}
\item{Subcase $I>0$

We will show that $S_2\rightarrow\infty$ as $N\rightarrow\infty$.  Break up $S_2$ of Eq.~\eqref{def:S2} so that $S_2 = S_3+S_4$, where
\begin{align}
S_3&=\sum_{k=\lfloor \ln{N}\rfloor+1}^{\lfloor Nx^* \rfloor-1}
\prod_{i=1}^k f\left(\frac{i}{N},N\right),\nonumber\\
S_4&=\sum_{k=\lfloor Nx^* \rfloor}^{N-1}
\prod_{i=1}^k f\left(\frac{i}{N},N\right),\label{def:S4}
\end{align}
and $x^*$ is defined as the point where $\int_0^{x^*} \ln\tilde{f}\left(x\right)dx=0$ (see Fig.~\ref{fig:integral}).
 Define
\begin{align}
m_4 &= \min_{\lfloor Nx^* \rfloor \le i \le N-1}\left\{f\left(\frac{i}{N},N\right)\right\}.\label{eq:f4_min}
\end{align}
Since $\tilde{f}$ is increasing, $m_4$ has the limit: $\lim_{N\rightarrow\infty} m_4 = \tilde{f}\left(x^*\right)>1$.
This implies the inequality:
\begin{align}
S_4 &\ge \sum_{k=\lfloor Nx^* \rfloor }^{N-1} \prod_{i=1}^k m_4 = \sum_{k=\lfloor Nx^* \rfloor }^{N-1} m_4^k = \frac{m_4^{\lfloor Nx^* \rfloor }-m_4^{N}}{1-m_4}.\label{ineq:S4}
\end{align}
Therefore, $\lim_{N\rightarrow\infty}S_4 =\infty$, which implies 
that $\lim_{N\rightarrow\infty}S = \infty$ and
$\lim_{N\rightarrow\infty} \rho_A = 0$.
}

\item{Subcase $I = 0$

We will show that limit of $S_2$ as $N\rightarrow\infty$ is positive and finite.
Let $\hat{x} = (b-d)/(b-d+c-a)$ be the point for which $\tilde f\left(\hat{x}\right) = 1$ (see Fig.~\ref{fig:integral}). Consider a sequence $\beta_N$ that satisfies
\begin{align*}
\hat{x} < \lim_{N\rightarrow\infty} \frac{\beta_N}{N} < 1,
\end{align*}
and converges to a limit $\beta=\lim_{N\rightarrow\infty} \beta_N/N$.  We will split $S_2$ of Eq.~\eqref{def:S2} at $k = \beta_N$, such that $S_2 = S_5+S_6$, where $S_6$ is the right tail-end of the sum.  We will show that $S_5\rightarrow 0$ and $S_6$ approaches a positive constant as $N\rightarrow\infty$. 
Set 
\begin{align}
S_5&=\sum_{k=\lfloor \ln{N}\rfloor+1}^{\beta_N-1}
\prod_{i=1}^k f\left(\frac{i}{N},N\right)\nonumber\\
S_6&=\sum_{k=\beta_N}^{N-1}
\prod_{i=1}^k f\left(\frac{i}{N},N\right).\label{def:S6}
\end{align}
To obtain the limit of $S_5$ we define
\begin{align*}
\tilde{S}_5 = \sum_{k=\lfloor \ln{N}\rfloor+1}^{\beta_N-1}\prod_{i=1}^k \exp \tilde{A}_k,
\end{align*}
where $\tilde{A}_k$ is given in Eq.~\eqref{def:Ak}.  Set $C = \int_0^{\beta}\ln\tilde{f}(x) \; dx$.  Importantly, $C<0$ since $I=0$ and $\ln\tilde{f}(x)$ is monotonic.
Similar arguments as in case \ref{case:intneg} show that
\begin{equation*}
\tilde{S}_5<
\sum_{k=\lfloor \ln{N}\rfloor+1}^{\beta_N-1} \frac{\tilde{f}\left(1\right)}{\tilde{f}\left(0\right)} \, e^{kC} =\frac{\tilde{f}\left(1\right)}{\tilde{f}\left(0\right)}\cdot\frac{e^{C(\lfloor \ln{N}\rfloor+1)}-e^{C\beta_N}}{1-e^{C}}.
\end{equation*}
Since $C<0$, it follows that
\begin{align}
\lim_{N\rightarrow\infty}\tilde{S}_5&= 0.\label{eq:S5tilde_limit}
\end{align}
To relate $\tilde{S}_5$ to $S_5$, we substitute $\tilde{m} = d/b$ into Eq.~\eqref{ineq:ftilde_prod} to obtain an upper bound for $S_5$,
\begin{align*}
S_5 &\le \left(1+\frac{bL}{dN}\right)^N \tilde{S}_5.
\end{align*}
Consequently, from Eqs.~\eqref{eq:exp_limit} and \eqref{eq:S5tilde_limit}, 
\begin{align}\label{eq:S5_limit}
\lim_{N\rightarrow\infty}S_5 = 0.
\end{align}

We now turn our attention to $S_6$ of Eq.~\eqref{def:S6}.
Define 
\begin{align*}
m_6 &= \min_{\beta_N \le i \le N-1}\left\{f\left(\frac{i}{N},N\right)\right\},\\
M_6 &= \max_{\beta_N \le i \le N-1}\left\{f\left(\frac{i}{N},N\right)\right\},
\end{align*}
which have the limits  $\lim_{N\rightarrow\infty} m_6 = \tilde{f}(\beta) >  1$ and $\lim_{N\rightarrow\infty} M_6= \tilde{f}\left(1\right)=c/a>1$.
Rewrite $S_6$ as
\begin{align}
S_6&=\left[\prod_{i=1}^{N-1} f\left(\frac{i}{N},N\right)\right]\left[1+\sum_{k=\beta_N}^{N-2}
\prod_{j=k+1}^{N-1} \left(f\left(\frac{j}{N},N\right)\right)^{-1}\right]\nonumber\\
&=\left[\prod_{i=1}^{N-1} f\left(\frac{i}{N},N\right)\right]\left[1+\sum_{\ell=1}^{N-\beta_N-1}
\prod_{h=1}^{\ell} \left(f\left(\frac{N-h}{N},N\right)\right)^{-1}\right].\label{eq:S6_prod}
\end{align}

Denote the second factor on the right-hand side of Eq.~\eqref{eq:S6_prod} by $\hat{S}_6$.  We have the bounds
\begin{align*}
1+\sum_{\ell=1}^{N-\beta_N-1} M_6^{-\ell}&\le \hat{S}_6\le 1+\sum_{\ell=1}^{N-\beta_N-1}m_6^{-\ell}\\
1+\frac{M_6^{-N+\beta_N}-M_6^{-1}}{M_6^{-1}-1}&\le \hat{S}_6 \le 1+\frac{m_6^{-N+\beta_N}-m_6^{-1}}{m_6^{-1}-1}.
\end{align*}
Now taking $N\rightarrow\infty$,
\begin{align}
\lim_{N\rightarrow\infty}\frac{M_6}{M_6-1}&\le\lim_{N\rightarrow\infty} \hat{S}_6 \le\lim_{N\rightarrow\infty}\frac{m_6}{m_6-1}\nonumber\\
\frac{\tilde{f}(1)}{\tilde{f}(1)-1}&\le\lim_{N\rightarrow\infty} \hat{S}_6 \le\frac{\tilde{f}(\beta)}{\tilde{f}(\beta)-1}.\label{eq:S6hat_squeeze}
\end{align} 
Since Eq.~\eqref{eq:S6hat_squeeze} is true for all $\beta$ with $\hat{x}<\beta<1$, then
\begin{align}
\lim_{N\rightarrow\infty} \hat{S}_6 = \frac{\tilde{f}(1)}{\tilde{f}(1)-1}=\frac{c}{c-a}.\label{eq:S6hat}
\end{align}

To analyze the first factor of Eq.~\eqref{eq:S6_prod}, we look at the version with $\tilde{f}$, which we relate to the integral $I$. Apply the Extended Trapezoidal Rule \citep{abramowitz1964handbook} to $I$: 
\begin{align*}
\int_0^1\ln\tilde{f}(x)dx = \frac{1}{N}\left[\frac{\ln\tilde{f}(0)+\ln\tilde{f}(1)}{2}+\sum_{i=1}^{N-1}\ln\tilde{f}\left(\frac{i}{N}\right)\right]+\mathcal{O}\left(N^{-2}\right).
\end{align*}
Recalling that $I=0$, $\tilde{f}(0)=d/b$ and $\tilde{f}(1)=c/a$, we obtain the asymptotic expansion:
\begin{align}
\sum_{i=1}^{N-1}\ln\tilde{f}\left(\frac{i}{N}\right)&=\ln\sqrt{\frac{ab}{cd}}+\mathcal{O}(N^{-1}).\label{eq:A_N}
\end{align}

To compare the sum in Eq.~\eqref{eq:A_N} with $\sum_{i=1}^{N-1}\ln f\left(\frac{i}{N},N\right)$, we look at their  difference:

\begin{align*}
\sum_{i=1}^{N-1} &\ln f\left(\frac{i}{N},N\right) - \sum_{i=1}^{N-1} \ln \tilde{f}\left(\frac{i}{N}\right)\\
&=\sum_{i=1}^{N-1} \ln\frac{f\left(\frac{i}{N},N\right)}{\tilde{f}\left(\frac{i}{N}\right)}\\
&= \sum_{i=1}^{N-1} \ln \left[1+\frac{1}{N}\left(\frac{a}{b+\frac{i}{N}(a-b)-\frac{a}{N}}-\frac{d}{d+\frac{i}{N}(c-d)}\right)\right.\\
&\hspace{2.1cm}\left.-\frac{1}{N^2}\frac{ad}{\left(d+\frac{i}{N}(c-d)\right)\left(b+\frac{i}{N}(a-b)-\frac{a}{N}\right)}\right].
\end{align*}
As $N \to \infty$, we have the asymptotic expression
\begin{align*}
\sum_{i=1}^{N-1} &\ln f\left(\frac{i}{N},N\right) - \sum_{i=1}^{N-1} \ln \tilde{f}\left(\frac{i}{N}\right)\nonumber\\
&=\frac{1}{N}\sum_{i=1}^{N-1}\left(\frac{a}{b+\frac{i}{N}(a-b)}-\frac{d}{d+\frac{i}{N}(c-d)}\right)+\mathcal{O}(N^{-1}).
\end{align*}
If we add and subtract $(a-b)/(bN)$ to the right-hand side, we obtain a left Riemann sum, which can be replaced as $N \to \infty$ by an integral:
\begin{align}
\sum_{i=1}^{N-1} &\ln f\left(\frac{i}{N},N\right) - \sum_{i=1}^{N-1} \ln \tilde{f}\left(\frac{i}{N}\right)\nonumber\\
& = \frac{1}{N}\sum_{i=0}^{N-1} \left(\frac{a}{b+\frac{i}{N}(a-b)}-\frac{d}{d+\frac{i}{N}(c-d)}\right) 
-\frac{a-b}{bN} +\mathcal{O}(N^{-1}) \nonumber\\
&=\int_0^1 \left(\frac{a}{b+x(a-b)}-\frac{d}{d+x(c-d)}\right)dx +\mathcal{O}(N^{-1})\nonumber\\
&=\ln\left(\frac{a^{\frac{a}{a-b}}d^{\frac{d}{c-d}}}{b^{\frac{a}{a-b}}c^{\frac{d}{c-d}}}\right)+\mathcal{O}(N^{-1}).\label{eq:diff_of_sum_largeN_3}
\end{align}
The logarithm can be simplified using the condition $I=0$.  Eq.~\eqref{eq:int_value} gives $a^{\frac{a}{a-b}}d^{\frac{d}{c-d}}=b^{\frac{b}{a-b}}c^{\frac{c}{c-d}}$, therefore
\begin{align*}
\ln\left(\frac{a^{\frac{a}{a-b}}d^{\frac{d}{c-d}}}{b^{\frac{a}{a-b}}c^{\frac{d}{c-d}}}\right)&=\ln\left(\frac{b^{\frac{b}{a-b}}c^{\frac{c}{c-d}}}{b^{\frac{a}{a-b}}c^{\frac{d}{c-d}}}\right)=\ln\left(\frac{c}{b}\right).
\end{align*}
Eq.~\eqref{eq:diff_of_sum_largeN_3} then simplifies to
\begin{equation}
\sum_{i=1}^{N-1} \ln f\left(\frac{i}{N},N\right) - \sum_{i=1}^{N-1} \ln \tilde{f}\left(\frac{i}{N}\right)=\ln\left(\frac{c}{b}\right)+\mathcal{O}(N^{-1}).\label{eq:diff_of_sum_largeN_2}
\end{equation}
Combining Eqs.~\eqref{eq:A_N} and \eqref{eq:diff_of_sum_largeN_2} yields
\begin{align*}
\sum_{i=1}^{N-1} \ln f\left(\frac{i}{N},N\right) &= \ln\sqrt{\frac{ab}{cd}}+\ln\left(\frac{c}{b}\right)+\mathcal{O}(N^{-1})\nonumber\\
&=\ln\sqrt{\frac{ac}{bd}}+\mathcal{O}(N^{-1}).
\end{align*}
Thus, $\prod_{i=1}^{N-1}f\left(\frac{i}{N},N\right) =\sqrt{ac/(bd)}+\mathcal{O}(N^{-1})$
and
\begin{align}
\lim_{N\rightarrow\infty}\prod_{i=1}^{N-1}f\left(\frac{i}{N},N\right) &=\sqrt{\frac{ac}{bd}}.\label{eq:A_N_limit}
\end{align}

Combine Eqs.~\eqref{eq:S6hat} and \eqref{eq:A_N_limit} with \eqref{eq:S6_prod} to obtain
\begin{align}\label{ineq:S6_limit}
\lim_{N\rightarrow\infty} S_6 =\frac{c}{c-a}\sqrt{\frac{ac}{bd}}.
\end{align}
Altogether Eqs.~\eqref{eq:S1_limit}, \eqref{eq:S5_limit} and \eqref{ineq:S6_limit} give
\begin{align*}
\lim_{N\rightarrow\infty} S  = \lim_{N\rightarrow\infty} (S_1 + S_5 + S_6) = \frac{d}{b-d} +\frac{c}{c-a}\sqrt{\frac{ac}{bd}},
\end{align*}
and from Eq.~\eqref{eq:rhoA_S},
\begin{equation}
\label{eq:rhoA_NwI0}
\lim_{N\rightarrow\infty} \rho_A =
\frac{(b-d)(c-a)}{b(c-a)+c(b-d)\sqrt{\frac{ac}{bd}}}. 
\end{equation}
}
\end{enumerate}
}
\item{Subcase $a = c$

In this case, $\tilde{f}$ is a strictly increasing function with minimum value $\tilde{f}(0) = d/b <1$ and maximum value  $\tilde{f}(1) = c/a =1$.  Thus, $\ln\tilde{f}(x) < 0$ for all $x\in[0,1)$ implying that $I<0$.  The same argument used in the case \ref{case:intneg} applies here.  We obtain the result $\lim_{N\rightarrow\infty}\rho_A = (b-d)/b$. \qed
}
\end{enumerate}}
\end{enumerate}
\end{proof}

Theorem \ref{thm:rhoA_Nlimit} gives the large-population limit of $\rho_A$. We now introduce weak selection to obtain asymptotic expressions for  $\rho_A$ in the \Nw limit.

 \begin{corollary} 
In the \Nw limit,
 \begin{align*}
\rho_A \sim
      \begin{cases}
   o(w) & \text{if } b\le d\\
   o(w) & \text{if } b>d \text{ and } a+b<c+d\\
      \frac{b-d}{2}w +o(w) & \text{if } b>d \text{ and } a+b=c+d\\
      (b-d)w +o(w) & \text{if } b>d \text{ and } a+b>c+d.
      \end{cases}
\end{align*} 
\end{corollary}

\begin{proof}
We introduce weak selection according to Eq.~\eqref{eq:weakdef}.  The integral $I$, given in closed form in Eq.~\eqref{eq:int_value}, has the following expansion as $w \to 0$:
\begin{align}\label{eq:int_ftilde_weak}
I &= \frac{w}{2}\left(c-a+d-b\right)+\mathcal{O}(w^2).
\end{align}
We now separate into the cases of Theorem \ref{thm:rhoA_Nlimit}.
\begin{enumerate}

\item Case $(b\le d) \text{ or } \left(b>d, \, a<c \text{ and } I>0\right)$

First note that given the expansion of Eq.~\eqref{eq:int_ftilde_weak}, the condition $(b>d) \,\land (a<c) \,\land\, (I>0)$ is equivalent to $(b>d) \land (a+b<c+d)$. Since $\lim_{N\rightarrow \infty}\rho_A=0$, $\rho_A \sim o(w)$ by Definition~\ref{def:Nw_totes}.

\item Case $\left(b>d \text{ and } a\ge c\right) \text{ or }\left(b>d, \, a< c \text{ and } I<0 \right)$

Using Eq.~\eqref{eq:int_ftilde_weak}, these two conditions are described by one condition under weak selection:  $(b>d) \land (a+b>c+d)$. Apply weak selection to $(b-d)/b$ and take $N\rightarrow\infty$ to get 
\[\lim_{N\rightarrow \infty}\rho_A=\frac{w(b-d)}{1+wb}= w(b-d) + wR(w),\] 
where $\lim_{w\rightarrow 0} R(w) = 0$.   By Definition~\ref{def:Nw_totes}, $\rho_A \sim (b-d)w +o(w)$. 

\item Case $b>d,  a< c \text{ and } I=0 $

Given Eq.~\eqref{eq:int_ftilde_weak}, this case under weak selection is equivalent to $(b>d)\land (a+b=c+d)$. In particular, we have $b-d=c-a$, which allows the cancellation of a factor of $b-d$ from the numerator and denominator of Eq.~\eqref{eq:rhoA_NwI0}.  Applying weak selection and taking $N\rightarrow \infty$ yields
\[
\lim_{N\rightarrow \infty}\rho_A =\frac{b-d}{2}w+wR(w),
\]
where $\lim_{w\rightarrow 0} R(w) = 0$.  By Definition~\ref{def:Nw_totes},  $\rho_A \sim \frac{b-d}{2}w +o(w)$. \qed
\end{enumerate}
\end{proof}

\subsubsection{Conditions for Success}
To determine conditions for success ($\rho_A>1/N$ and $\rho_A>\rho_B$) in the \Nw limit, we must first determine such conditions in the limit of large population size.  To do so, we note that

\begin{theorem}\label{thm:N_NrhoA} 
$\rho_A > 1/N$ for sufficiently large $N$ if and only if one of the following holds:
\begin{enumerate}[label=(\roman*)]
\item \label{N_NrhoA_cond1} $b>d$ and $a\ge c$ 
\item \label{N_NrhoA_cond2} $b>d$, $a<c$ and $I \le 0$
\item \label{N_NrhoA_cond4} $b=d$ and $a>c$ 
\item \label{N_NrhoA_cond3} $b=d$, $a=c$ and $b>c$ 
\end{enumerate}
\end{theorem}

\begin{proof}
 $\rho_A > 1/N$ for sufficiently large $N$ if $\lim_{N\rightarrow\infty}N\rho_A > 1$. From Eq.~\eqref{eq:rhoA_S}, we have the relation 
 \begin{align}\label{eq:NrhoA}
\lim_{N\rightarrow\infty} N\rho_A = \lim_{N\rightarrow\infty} \frac{N}{1+S} 
= \lim_{N\rightarrow\infty} \left( \frac{1}{N}+\frac{S}{N} \right)^{-1} = \left(  \lim_{N\rightarrow\infty} \frac{S}{N} \right)^{-1} .
 \end{align}
If $\lim_{N\rightarrow\infty} S/N =0$, then $\lim_{N\rightarrow\infty} N\rho_A = \infty$ since $S$ is always positive.
We will consider $\lim_{N\rightarrow\infty}N\rho_A$ for the following cases.
\begin{enumerate}
	\item Case ($b >d$ and $a\ge c$) or ($b>d$, $a<c$ and $I \le 0$)
	
	From Eq.~\eqref{eq:rhoA_Nlimit}, $\lim_{N\rightarrow\infty}\rho_A$ is positive and finite.  Thus, $\lim_{N\rightarrow\infty}N\rho_A =\infty$.

\item Case $ b > d$, $a<c$ and $I>0$

Given $S \ge S_4$ and $\lim_{N\rightarrow\infty}m_4 =\tilde{f}(x^*)> 1$, where $S_4$ and $m_4$ are defined in Eqs.~\eqref{def:S4} and \eqref{eq:f4_min}, respectively, we use the inequality of Eq.~\eqref{ineq:S4} to obtain
\begin{align*}
\lim_{N\rightarrow\infty}\frac{S}{N} &\ge \lim_{N\rightarrow\infty}
\frac{m_4^{\lfloor Nx^* \rfloor }-m_4^{N}}{N(1-m_4)}=\infty.
\end{align*}
Therefore from Eq.~\eqref{eq:NrhoA}, $\lim_{N\rightarrow\infty}N\rho_A =0$.

\item Case $b<d$ 

Given Eq.~\eqref{ineq:S1_min_max},
\begin{align*}
S&\ge S_1\ge \sum_{k=1}^{\lfloor\ln N\rfloor}m_1^k = \frac{m_1^{\lfloor\ln N\rfloor+1}-m_1}{m_1-1}.
\end{align*}
Since $\lim_{N\rightarrow\infty} m_1 = \frac{d}{b} >1$,
\begin{align*}
\lim_{N\rightarrow\infty}\frac{S}{N}&\ge \lim_{N\rightarrow\infty}\frac{m_1^{\lfloor\ln N\rfloor+1}-m_1}{N(m_1-1)}=\infty.
\end{align*}
Therefore, $\lim_{N\rightarrow\infty}N\rho_A =0$.  
 
\item Case $b=d$ 
\begin{enumerate}
\item Subcase $a = b=c=d$

$f\left(\frac{i}{N},N\right) =1$ for all $i$ with $\rho_A = \frac{1}{N}$.  Therefore, $\lim_{N\rightarrow\infty}N\rho_A = 1$.

\item Subcase $a = c$ and $a\ne b $ 

Here
\begin{align*}
\frac{\partial f}{\partial x} &=-\frac{(a-b)^2}{N\left(a\left(\frac{1}{N}-x\right)+b(x-1)\right)^2}.
\end{align*}
Therefore, $f$ is a decreasing function.  Let 
\begin{align*}
\begin{split}
m &= \min_{1\le i \le N-1}\left\{f\left(\frac{i}{N},N\right)\right\} = f\left(\frac{N-1}{N},N\right)=\frac{a-\frac{a}{N}}{a+\frac{b-2a}{N}},\\
M &= \max_{1\le i \le N-1}\left\{f\left(\frac{i}{N},N\right)\right\}= f\left(\frac{1}{N},N\right)= \frac{b+\frac{a-2b}{N}}{b-\frac{b}{N}}.
\end{split}
\end{align*}
Then
\begin{align}
\sum_{k=1}^{N-1}m^k &\le S\le \sum_{k=1}^{N-1}M^k \nonumber\\
 \frac{m^N-m}{N(m-1)}&\le \frac{S}{N}\le \frac{M^N-M}{N(M-1)}\label{ineq:SoverN}
\end{align}

Note that $\lim_{N\rightarrow\infty} m =\lim_{N\rightarrow\infty} M = 1$. To determine the limit of $S/N$ as $N\rightarrow\infty$, we require the derivatives:
\begin{align}
\frac{dm}{dN} &= \frac{a(b-a)}{(Na+b-2a)^2},\label{eq:dmdN}\\
\frac{dM}{dN} &= \frac{b-a}{b(N-1)^2}\label{eq:dMdN}.
\end{align}
Applying  L'H\^opital's Rule and using Eqs.~\eqref{eq:dmdN} and \eqref{eq:dMdN}, we obtain the following limits:
\begin{align*}
\lim_{N\rightarrow\infty} N(m-1) &= \lim_{N\rightarrow\infty} \frac{\frac{dm}{dN}}{-N^{-2}}= \frac{a-b}{a},\\
\lim_{N\rightarrow\infty} N(M-1) &=\lim_{N\rightarrow\infty} \frac{\frac{dM}{dN}}{-N^{-2}}=\frac{a-b}{b},\\
\lim_{N\rightarrow\infty} N\ln m &= \lim_{N\rightarrow\infty} \frac{\frac{1}{m}\frac{dm}{dN}}{-N^{-2}}= \frac{a-b}{a}\\
\lim_{N\rightarrow\infty} N\ln M &= \lim_{N\rightarrow\infty} \frac{\frac{1}{M}\frac{dM}{dN}}{-N^{-2}}= \frac{a-b}{b}
\end{align*}
Therefore, 
\begin{align*}
\begin{split}
\lim_{N\rightarrow\infty} m^N &= \lim_{N\rightarrow\infty} e^{N\ln m }=\exp\left(\frac{a-b}{a}\right),\\
\lim_{N\rightarrow\infty} M^N &= \lim_{N\rightarrow\infty} e^{N\ln M }=\exp\left(\frac{a-b}{b}\right).
\end{split}
\end{align*}
Take the limit of Eq.~\eqref{ineq:SoverN} to obtain
\begin{align*}
\frac{\exp\left(\frac{a-b}{a}\right)-1}{\frac{a-b}{a}}&\le \lim_{N\rightarrow\infty}\frac{S}{N}\le \frac{\exp\left(\frac{a-b}{b}\right)-1}{\frac{a-b}{b}}.
\end{align*}
If $a>b$ (equivalently $b<c$) then $\lim_{N\rightarrow\infty}S/N>1$.
If $a<b$ (equivalently $b>c$) then $\lim_{N\rightarrow\infty}S/N<1$.  Thus, $\lim_{N\rightarrow\infty}N\rho_A>1$ if $b=d, a=c$ and $b>c$ by Eq.~\eqref{eq:NrhoA}.

\item Subcase $a < c$

Set
\begin{align*}
m_7 &= \min_{N-\lfloor\ln N\rfloor\le i\le N-1}\left\{f\left(\frac{i}{N},N\right)\right\}.
\end{align*}
Note that $\lim_{N\rightarrow\infty}m_7 = \tilde{f}(1)= c/a>1$ given that $f$ converges uniformly to $\tilde{f}$. Then
\begin{align*}
\lim_{N\rightarrow\infty} \frac{S}{N}&\ge \lim_{N\rightarrow\infty} \frac{1}{N}\sum_{k=N-\lfloor \ln N \rfloor }^{N-1}m_7^k\\
&=\lim_{N\rightarrow\infty}
\frac{m_7^N-m_7^{N-\lfloor \ln N \rfloor }}{N(m_7-1)}=\infty
\end{align*}
By Eq.~\eqref{eq:NrhoA}, $\lim_{N\rightarrow\infty}\rho_A =0$.

\item Subcase $a > c$

Here
\begin{align*}
\frac{\partial f}{\partial x} &=\frac{b(c-a)+\frac{2ab-b^2-ac}{N}}{N^2\left(b+x(a-b)-\frac{a}{N}\right)^2}.
\end{align*}
Therefore for $N > \frac{2ab-b^2-ac}{b(a-c)}, f$ is strictly decreasing.

Break up the sum $S$ as $S=S_8+S_9$, where
\begin{align*}
S_8 &= \sum_{k=1}^{\lfloor\sqrt{N}\rfloor-1} \prod_{i=1}^k f\left(\frac{i}{N},N\right)\\
S_9&=\sum_{k=\lfloor\sqrt{N}\rfloor}^{N-1} \prod_{i=1}^k f\left(\frac{i}{N},N\right).
\end{align*}
Given $N > \frac{2ab-b^2-ac}{b(a-c)}$, define
\begin{align*}
M_8 &= \max_{1\le i \le \lfloor\sqrt{N}\rfloor-1} f\left(\frac{i}{N},N\right) = 
f\left(\frac{1}{N},N\right) = \frac{b+\frac{c-2b}{N}}{b-\frac{b}{N}},\\
M_9 &=\max_{\lfloor\sqrt{N}\rfloor\le i \le N-1} f\left(\frac{i}{N},N\right) = 
f\left(\frac{\lfloor\sqrt{N}\rfloor}{N},N\right) = \frac{b+\frac{1}{\lfloor\sqrt{N}\rfloor}(c-b)-\frac{b}{N}}{b+\frac{1}{\lfloor\sqrt{N}\rfloor}(a-b)-\frac{a}{N}}.
\end{align*}
If $c = b$ then $M_8 = 1$ and we have the bound
\begin{align*}
S_8 \le \sum_{k=1}^{\lfloor\sqrt{N}\rfloor-1} M_8^k = \lfloor\sqrt{N}\rfloor-1.  
\end{align*}
Dividing by $N$ and taking $N\rightarrow\infty$ we obtain
\begin{align*}
\lim_{N\rightarrow\infty}\frac{S_8}{N} \le \frac{\lfloor\sqrt{N}\rfloor-1}{N} = 0.
\end{align*}

If $c \ne b$, we have the bound
\begin{align}
S_8 \le \sum_{k=1}^{\lfloor\sqrt{N}\rfloor-1} M_8^k = \frac{M_8^{\lfloor\sqrt{N}\rfloor}-M_8}{M_8-1}.\label{ineq:S8}
\end{align}
Note that $\lim_{N\rightarrow\infty} M_8 = 1$. We use  L'H\^opital's Rule to determine the limit of $S_8/N$ as $N\rightarrow\infty$, which requires the derivative: $\frac{dM_8}{dN} = \frac{b-c}{b(N-1)^2}$. It follows that
\begin{align*}
\lim_{N\rightarrow\infty} N(M_8-1) &=\lim_{N\rightarrow\infty} \frac{\frac{dM_8}{dN}}{-N^{-2}}= \frac{c-b}{b},\\
\lim_{N\rightarrow\infty} \lfloor\sqrt{N}\rfloor \ln M_8 &=\lim_{N\rightarrow\infty} \frac{\frac{1}{M_8}\frac{dM_8}{dN}}{-\frac{1}{2}N^{-3/2}} = 0.
\end{align*}
Therefore, $\lim_{N\rightarrow\infty}M_8^{\lfloor\sqrt{N}\rfloor}=1$, and consequently from Eq.~\eqref{ineq:S8},
\begin{align}\label{eq:S8_limit}
\lim_{N\rightarrow\infty}\frac{S_8}{N} \le\lim_{N\rightarrow\infty} \frac{M_8^{\lfloor\sqrt{N}\rfloor}-M_8}{N(M_8-1)}=0.
\end{align}

We also have an upper bound for $S_9$:
\begin{align*}
S_9&\le \sum_{k=\lfloor\sqrt{N}\rfloor}^{N-1} M_9^k \le \frac{1}{1-M_9}=\frac{b+\frac{a-b}{\lfloor\sqrt{N}\rfloor}-\frac{a}{N}}{\frac{a-c}{\lfloor\sqrt{N}\rfloor}+\frac{b-a}{N}}
\end{align*}
Divide by $N$ and take the $N\to \infty$ limit to obtain
\begin{align}\label{eq:S9_limit}
\lim_{N\rightarrow\infty}\frac{S_9}{N}&\le \lim_{N\rightarrow\infty} \frac{b+\frac{a-b}{\lfloor\sqrt{N}\rfloor}-\frac{a}{N}}{\sqrt{N}(a-c)+b-a}=0
\end{align}
Eqs.~\eqref{eq:S8_limit} and \eqref{eq:S9_limit} imply $\lim_{N\rightarrow\infty} S/N = 0$, and consequently $\lim_{N\rightarrow\infty}N\rho_A =\infty$.\qed
\end{enumerate}
\end{enumerate} 
\end{proof}

We now apply weak selection to find conditions for which $\rho_A>1/N$ in the \Nw limit.
\begin{corollary} 
\label{thm:Nw_NrhoA}
Given the game matrix \eqref{def:game},  $\rho_A > 1/N$ in the \Nw limit if and only if one of the following holds:
\begin{enumerate}[label=(\roman*)]
\item $b>d$ and $a+b \ge c+d$
\item $b=d$ and $a>c$
\item $b=d, a=c$ and $b>c$
\end{enumerate}
\end{corollary}

\begin{proof}
In Theorem \ref{thm:N_NrhoA}, we found conditions for which $\rho_A>1/N$ for sufficiently large populations. We  introduce weak selection according to Eq.~\eqref{eq:weakdef}. Conditions \ref{N_NrhoA_cond1}, \ref{N_NrhoA_cond4} and \ref{N_NrhoA_cond3} of Theorem \ref{thm:N_NrhoA} remain the same under weak selection.  
Given the weak selection expansion of Eq.~\eqref{eq:int_ftilde_weak}, Condition \ref{N_NrhoA_cond2} of Theorem \ref{thm:N_NrhoA} becomes $(b>d)\land(a<c)\land(a+b\ge c+d)$. 
Note that Condition \ref{N_NrhoA_cond1} of Theorem \ref{thm:N_NrhoA} is equivalent to $(b>d)\land(a\ge c)\land(a+b\ge c+d)$. Therefore, Conditions \ref{N_NrhoA_cond1} and \ref{N_NrhoA_cond2} of Theorem \ref{thm:N_NrhoA} together give the one condition $(b>d)\land(a+b\ge c+d)$.
\qed
\end{proof}

Finally, we will determine conditions for which $\rho_A > \rho_B$ in the \Nw limit by first investigating the large $N$ limit.

\begin{theorem} 
\label{thm:ratio_largeN} 
Given the game matrix ~\eqref{def:game},  $\rho_A > \rho_B$ for sufficiently large $N$ if and only if one of the following conditions holds:
\begin{enumerate}[label=(\roman*)]
\item \label{case:ratio_Ineg} $I<0$
\item \label{case:ratio_I0} $I=0$ and  $ac<bd$
\end{enumerate}
\end{theorem}

\begin{proof}
Eq.~\eqref{eq:Moran_ratio} with Eq.~\eqref{def:f} give
\begin{align}
\frac{\rho_B}{\rho_A} &= \prod_{i=1}^{N-1}f\left(\frac{i}{N},N\right) \label{eq:rhoB_A_prodf}
\end{align}
Given that $\rho_A > \rho_B$ for sufficiently large $N$ if and only if $\lim_{N\rightarrow\infty}\rho_B/\rho_A < 1$, 
we will find this limit and compare it to 1 for various cases.

We will first look at the product of $\tilde{f}$-terms and then compare it to the product of $f$-terms. Since $\tilde{f}$ is monotonic, the left and right Riemann sums, $\frac{1}{N}\left(\tilde{A}_{N-1} + \ln \tilde{f}(0)\right)$ and $\frac{1}{N}\left(\tilde{A}_{N-1} + \ln \tilde{f}(1)\right)$, respectively, serve as bounds for the definite integral $I$ (where $\tilde{A}_{N-1}$ is defined in Eq.~\eqref{def:Ak}).  This implies
\begin{align}
NI -\ln\tilde{M}&\le\tilde{A}_{N-1}\le N I -\ln\tilde{m},\label{ineq:A_N_RS}
\end{align}
where the minimum, $\tilde{m}$, and maximum, $\tilde{M}$, of $\tilde{f}$ are defined in Eq.~\eqref{def:mtilde}.
Keeping in mind that $\prod_{i=1}^{N-1}\tilde{f}\left(\frac{i}{N}\right) = \exp\left(\tilde{A}_{N-1}\right)$, exponentiate Eq.~\eqref{ineq:A_N_RS} to obtain
\begin{align*}
\frac{e^{NI}}{\tilde{M}}&\le\prod_{i=1}^{N-1}\tilde{f}\left(\frac{i}{N}\right)\le\frac{e^{NI}}{\tilde{m}}.
\end{align*}
Combining this with the inequality of Eq.~\eqref{ineq:ftilde_prod}, which compares $f$ to $\tilde{f}$, and using Eq.~\eqref{eq:rhoB_A_prodf}, we obtain
\begin{align*}
\left(1-\frac{L}{\tilde{m}N}\right)^{N-1}\frac{e^{NI}}{\tilde{M}}&\le\frac{\rho_B}{\rho_A}\le\left(1+\frac{L}{\tilde{m}N}\right)^{N-1}\frac{e^{NI}}{\tilde{m}}.
\end{align*}
Thus, if $I>0$ then $\lim_{N\rightarrow\infty}\rho_B/\rho_A = \infty$.  If $I<0$ then $\lim_{N\rightarrow\infty}\rho_B/\rho_A = 0$.  
The only case left to consider is $I=0$. In this case, Eq.~\eqref{eq:A_N_limit} implies that $\lim_{N\rightarrow\infty}\rho_B/\rho_A=\sqrt{ac/(bd)}$.
If $ac>bd$ then the limit is greater than 1, if $ac<bd$ then the limit is less than 1, and if $ac=bd$ then the limit equals 1.
\qed
\end{proof}

\begin{corollary}
\renewcommand{\labelenumi}{(\roman{enumi})}
  Given the game matrix \eqref{def:game}, $\rho_A > \rho_B$ in the \Nw limit if and only if one of the following holds:
\begin{enumerate}
	\item $a+b>c+d$
	\item $a+b=c+d$ and $b>c$
\end{enumerate}
\end{corollary}

\begin{proof}
We introduce weak selection according to Eq.~\eqref{eq:weakdef}. Given the weak selection expansion of integral $I$ in Eq.~\eqref{eq:int_ftilde_weak},   $I<0$ implies $a+b>c+d$ and $I=0$ implies $a+b=c+d$.  Furthermore, the inequality $ac<bd$  is
\[
(1+wa)(1+wc)<(1+wb)(1+wd),
\]
which reduces as $w\to 0$ to
\begin{equation}
a+c<b+d.\label{ineq:ratio_condweak}
\end{equation}
Thus, Condition \ref{case:ratio_Ineg} of Theorem \ref{thm:ratio_largeN} becomes $a+b>c+d$ and Condition
\ref{case:ratio_I0}  becomes  $a+b=c+d$ and $b>c$ (equivalently $a+b=c+d$ and $a<d$) in the \Nw limit. \qed
\end{proof}

\section{Discussion}
In the analysis of evolutionary models, the limits of large population size and weak selection are both biologically important and mathematically convenient.  We have analyzed the effect of combining these limits, in different orders, on the fixation of strategies in the Moran process with frequency dependence.  We find that the \Nw and \wN limits yield different asymptotic expressions for fixation probability, as well as different conditions for a strategy to have larger fixation probability than a neutral mutation.  Interestingly, however, the conditions are the same for $\rho_A > \rho_B$.

Our results connect to a number of concepts in evolutionary game theory and population genetics.  For example, the conditions for $\rho_A > 1/N$ in the \Nw limit (Corollary \ref{thm:Nw_NrhoA}) have interesting connections to notions of evolutionary stability and risk dominance.  $A$ is an \emph{evolutionary stable strategy} (ESS) if $a>c$ or if $a=c$ and $b>d$; correspondingly, $B$ is an ESS if $d>b$ or if $d=b$ and $c>a$ \citep{MaynardSmith}. $A$ is \emph{risk dominant} if $a+b > c+d$, and $B$ is risk dominant if the reverse inequality holds \citep{harsanyi1988general,NowakFinite}. Comparing with Corollary \ref{thm:Nw_NrhoA}, we see that in the \Nw limit,
\[
\rho_A > 1/N \quad \Longrightarrow \quad \text{$B$ is neither an ESS nor risk dominant}.
\]
The converse holds outside of borderline cases where $a+b=c+d$.

In the \wN limit, we find in Theorem \ref{thm:wN_NrhoA} (see also \citealp{bomze2008one}) that $a+2b>c+2d$ is sufficient for $\rho_A > 1/N$, and is necessary except in the borderline case $a+2b=c+2d$.  This result is an instance of the one-third law of evolutionary game theory \citep{NowakFinite,ohtsuki2007one,bomze2008one,ladret2008evolutionary,zheng2011diffusion}.  This rule can be understood as stating that the conditions $\tilde{f}(1/3)>1$ and $\rho_A > 1/N$ are equivalent up to borderline cases. There does not appear to be any corresponding result for the \Nw limit; thus the one-third law appears to pertain specifically to the \wN limit.

The conditions for $\rho_A > \rho_B$, which are the same for the \Nw and \wN limits, are nearly equivalent to risk dominance, in that
\[
\text{$A$ is risk dominant} \quad \Longrightarrow \quad \rho_A > \rho_B,
\]
and the converse holds except in the borderline case $a+b=c+d$.  

Our analysis of the \Nw limit required us to first examine the large-population limit of $\rho_A$.  Here our results formalize and strengthen those of \cite{antal2006fixation}, who analyzed the same limit but used approximations that are not asymptotically exact in all cases.  Our results in Theorem \ref{thm:rhoA_Nlimit} confirm those of \cite{antal2006fixation} except in the borderline case $b>d, a<c$ and $I=0$.  In that case we have shown that
\[
\lim_{N\rightarrow \infty}\rho_A = \frac{(b-d)(c-a)}{b(c-a)+c(b-d)\sqrt{\frac{ac}{bd}}},
\]
which contradicts Antal and Scheuring's claim that $\lim_{N\rightarrow\infty}\rho_A=(b-d)/2b$.  These expressions are not equivalent; for example, they differ for the payoff matrix
\[
\begin{pmatrix}
e & 2e \\
4 & 4
\end{pmatrix},
\]
which satisfies the conditions $b>d$, $a<c$, and $I=0$. We can trace the difference to Antal and Scheuring's replacement of the sum $\tilde{A}_k$, defined in our Eq.~\eqref{eq:S2tilde}, by its integral approximation.

Here we have analyzed the \wN and \Nw limits for the Moran model of a well-mixed population with overlapping generations.  These limits can also be applied to other processes, where they may lead to novel questions or shed new light on existing results.  One interesting example is the Wright-Fisher model \citep{fisher1930genetical,wright1931evolution}, in which generations are non-overlapping.  In the case of a constant selection coefficient $s>0$, \cite{haldane1927mathematical} obtained the well-known approximation $\rho \approx 2s$.  We expect that this approximation will be asymptotically exact in the \Nw limit; for the \wN limit we anticipate a different result of the form $\rho \sim 1/N + Ks + o(s)$ for some positive coefficient $K$.  These limits can also be studied for Wright-Fisher model with games \citep{imhof2006evolutionary}, leading to discrete-generation analogues of the results presented here.  

Games on graphs \citep{Ohtsuki,SzaboFath,allen2014games} represent another important application.  For the case of the cycle \citep{OhtsukiCycles}, the \wN and \Nw limits were studied by \cite{jeong2014optional}, although without formal definitions and without considering borderline cases.  For regular graphs, \cite{Ohtsuki} obtained results that appear to pertain to the \wN limit; finite-$N$ corrections to these were later developed by \cite{Taylor} and \cite{chen2013sharp}.  It is not clear whether the \Nw limit is tractable for games on general graphs.  \cite{ibsen2015computational} showed that, for arbitrary graphs and nonweak selection, the problem of determining fixation probability is PSPACE-hard.  It is therefore very difficult to analyze evolutionary games on graphs without taking the weak selection limit at the outset.  However, this does not necessarily preclude computationally feasible conditions for success in the \Nw limit for at least some classes of graphs.

Finally, we reiterate that the \Nw and \wN limits represent only two of infinitely many ways to combine the large-population and weak-selection limits.  In the most general case, one considers an arbitrary sequence of pairs $\{(w_j, N_j)\}_{j=1}^\infty$ such that $w_j \to 0$ and $N_j \to \infty$ as $j \to \infty$.  It is clear from our results that expressions for fixation probability and conditions for success will depend on the sequence in question.  The \Nw and \wN limits represent two extremes in which one limit is taken much faster than the other.  It may be supposed that results for other limiting schemes will lie between these extremes in some sense.


\begin{thebibliography}{32}
\providecommand{\natexlab}[1]{#1}
\providecommand{\url}[1]{{#1}}
\providecommand{\urlprefix}{URL }
\expandafter\ifx\csname urlstyle\endcsname\relax
  \providecommand{\doi}[1]{DOI~\discretionary{}{}{}#1}\else
  \providecommand{\doi}{DOI~\discretionary{}{}{}\begingroup
  \urlstyle{rm}\Url}\fi

\bibitem[{Abramowitz and Stegun(1964)}]{abramowitz1964handbook}
Abramowitz M, Stegun IA (1964) Handbook of mathematical functions: with
  formulas, graphs, and mathematical tables, vol~55. Courier Corporation

\bibitem[{Allen and Nowak(2014)}]{allen2014games}
Allen B, Nowak MA (2014) Games on graphs. EMS Surveys in Mathematical Sciences
  1(1):113--151

\bibitem[{Antal and Scheuring(2006)}]{antal2006fixation}
Antal T, Scheuring I (2006) Fixation of strategies for an evolutionary game in
  finite populations. Bulletin of Mathematical Biology 68(8):1923--1944

\bibitem[{Blume(1993)}]{blume1993statistical}
Blume LE (1993) The statistical mechanics of strategic interaction. Games and
  economic behavior 5(3):387--424

\bibitem[{Bomze and Pawlowitsch(2008)}]{bomze2008one}
Bomze I, Pawlowitsch C (2008) One-third rules with equality: Second-order
  evolutionary stability conditions in finite populations. Journal of
  theoretical biology 254(3):616--620

\bibitem[{Broom and Rycht{\'a}r(2013)}]{broom2013game}
Broom M, Rycht{\'a}r J (2013) Game-Theoretical Models in Biology. Chapman \&
  Hall/CRC, Boca Raton, FL, USA

\bibitem[{Chen(2013)}]{chen2013sharp}
Chen YT (2013) Sharp benefit-to-cost rules for the evolution of cooperation on
  regular graphs. The Annals of Applied Probability 23(2):637--664

\bibitem[{Fisher(1930)}]{fisher1930genetical}
Fisher RA (1930) The Genetical Theory of Natural Selection. Oxford University
  Press

\bibitem[{Haldane(1927)}]{haldane1927mathematical}
Haldane JBS (1927) A mathematical theory of natural and artificial selection,
  part v: selection and mutation. In: Mathematical Proceedings of the Cambridge
  Philosophical Society, Cambridge Univ Press, vol~23, pp 838--844

\bibitem[{Harsanyi et~al(1988)Harsanyi, Selten et~al}]{harsanyi1988general}
Harsanyi JC, Selten R, et~al (1988) A general theory of equilibrium selection
  in games. MIT Press Books 1

\bibitem[{Hofbauer and Sigmund(1998)}]{Hofbauer1998}
Hofbauer J, Sigmund K (1998) {Evolutionary Games and Replicator Dynamics}.
  Cambridge University Press, Cambridge, UK

\bibitem[{Ibsen-Jensen et~al(2015)Ibsen-Jensen, Chatterjee, and
  Nowak}]{ibsen2015computational}
Ibsen-Jensen R, Chatterjee K, Nowak MA (2015) Computational complexity of
  ecological and evolutionary spatial dynamics. Proceedings of the National
  Academy of Sciences 112(51):15,636--15,641

\bibitem[{Imhof and Nowak(2006)}]{imhof2006evolutionary}
Imhof LA, Nowak MA (2006) Evolutionary game dynamics in a wright-fisher
  process. Journal of Mathematical Biology 52(5):667--681

\bibitem[{Jeong et~al(2014)Jeong, Oh, Allen, and Nowak}]{jeong2014optional}
Jeong HC, Oh SY, Allen B, Nowak MA (2014) Optional games on cycles and complete
  graphs. Journal of theoretical biology 356:98--112

\bibitem[{Ladret and Lessard(2008)}]{ladret2008evolutionary}
Ladret V, Lessard S (2008) Evolutionary game dynamics in a finite asymmetric
  two-deme population and emergence of cooperation. Journal of theoretical
  biology 255(1):137--151

\bibitem[{Lessard and Ladret(2007)}]{LessardFixation}
Lessard S, Ladret V (2007) {The probability of fixation of a single mutant in
  an exchangeable selection model}. Journal of Mathematical Biology
  54(5):721--744

\bibitem[{Maynard~Smith(1982)}]{MaynardSmith2}
Maynard~Smith J (1982) {Evolution and the Theory of Games}. Cambridge
  University Press, Cambridge

\bibitem[{Maynard~Smith and Price(1973)}]{MaynardSmith}
Maynard~Smith J, Price GR (1973) {The logic of animal conflict}. Nature
  246(5427):15--18

\bibitem[{Moran(1958)}]{Moran}
Moran PAP (1958) Random processes in genetics. Mathematical Proceedings of the
  Cambridge Philosophical Society 54(01):60--71

\bibitem[{Nowak and May(1992)}]{NowakMay}
Nowak MA, May RM (1992) {Evolutionary games and spatial chaos}. Nature
  359(6398):826--829

\bibitem[{Nowak et~al(2004)Nowak, Sasaki, Taylor, and Fudenberg}]{NowakFinite}
Nowak MA, Sasaki A, Taylor C, Fudenberg D (2004) {Emergence of cooperation and
  evolutionary stability in finite populations}. Nature 428(6983):646--650

\bibitem[{Nowak et~al(2010)Nowak, Tarnita, and Antal}]{NowakStructured}
Nowak MA, Tarnita CE, Antal T (2010) {Evolutionary dynamics in structured
  populations}. Philosophical Transactions of the Royal Society B: Biological
  Sciences 365(1537):19--30

\bibitem[{Ohtsuki and Nowak(2006)}]{OhtsukiCycles}
Ohtsuki H, Nowak MA (2006) Evolutionary games on cycles. Proceedings of the
  Royal Society B: Biological Sciences 273(1598):2249--2256,
  \doi{10.1098/rspb.2006.3576}
  
\bibitem[{Ohtsuki et~al(2006)Ohtsuki, Hauert, Lieberman, and Nowak}]{Ohtsuki}
Ohtsuki H, Hauert C, Lieberman E, Nowak MA (2006) {A simple rule for the
  evolution of cooperation on graphs and social networks}. Nature 441:502--505

\bibitem[{Ohtsuki et~al(2007)Ohtsuki, Bordalo, and Nowak}]{ohtsuki2007one}
Ohtsuki H, Bordalo P, Nowak MA (2007) The one-third law of evolutionary
  dynamics. Journal of theoretical biology 249(2):289--295

\bibitem[{Szab{\'o} and F{\'a}th(2007)}]{SzaboFath}
Szab{\'o} G, F{\'a}th G (2007) {Evolutionary games on graphs}. Physics Reports
  446(4-6):97--216

\bibitem[{Tarnita et~al(2009)Tarnita, Ohtsuki, Antal, Fu, and Nowak}]{Corina}
Tarnita CE, Ohtsuki H, Antal T, Fu F, Nowak MA (2009) Strategy selection in
  structured populations. Journal of Theoretical Biology 259(3):570 -- 581,
  \doi{10.1016/j.jtbi.2009.03.035}

\bibitem[{Taylor et~al(2004)Taylor, Fudenberg, Sasaki, and
  Nowak}]{TaylorFiniteGame}
Taylor C, Fudenberg D, Sasaki A, Nowak M (2004) Evolutionary game dynamics in
  finite populations. Bulletin of Mathematical Biology 66:1621--1644

\bibitem[{Taylor et~al(2007)Taylor, Day, and Wild}]{Taylor}
Taylor PD, Day T, Wild G (2007) {Evolution of cooperation in a finite
  homogeneous graph}. Nature 447(7143):469--472

\bibitem[{Weibull(1997)}]{weibull1997evolutionary}
Weibull JW (1997) Evolutionary Game Theory. MIT press, Cambridge, MA, USA

\bibitem[{Wright(1931)}]{wright1931evolution}
Wright S (1931) Evolution in mendelian populations. Genetics 16(2):97--159

\bibitem[{Zheng et~al(2011)Zheng, Cressman, and Tao}]{zheng2011diffusion}
Zheng X, Cressman R, Tao Y (2011) The diffusion approximation of stochastic
  evolutionary game dynamics: Mean effective fixation time and the significance
  of the one-third law. Dynamic Games and Applications 1(3):462--477

\end{thebibliography}
\end{document}